\documentclass[10pt,twocolumn,twoside]{IEEEtran}
\usepackage{}
\usepackage{amssymb}

\ifCLASSINFOpdf
   \usepackage[pdftex]{graphicx}
   \graphicspath{{../pdf/}{../jpeg/}}
   \DeclareGraphicsExtensions{.pdf,.jpeg,.png}
\else
   \usepackage[dvips]{graphicx}
   \graphicspath{{../eps/}}
   \DeclareGraphicsExtensions{.eps}
\fi

\usepackage{amsmath}

\usepackage{color}
\newtheorem{proposition}{Proposition}
\newtheorem{lemma}{Lemma}
\newtheorem{corollary}{Corollary}

\begin{document}

\title{Resource Allocation and Outage Analysis for An Adaptive Cognitive Two-Way Relay Network}

	\author{Qunwei~Li,~\IEEEmembership{Student Member,~IEEE},
		Pramod~K.~Varshney,~\IEEEmembership{Fellow,~IEEE}
		\thanks{This work was supported in part by the NSF under Grant ENG 1609916 and the AFOSR under Grant FA9550-16-1-0077.}
		\thanks{Q.\ Li and P.~K.\ Varshney are with the Department of Electrical Engineering and Computer Science, Syracuse University, Syracuse, NY 13244 USA
			(e-mail: qli33@syr.edu; varshney@syr.edu).}
}

\maketitle

\begin{abstract}
In this paper, an adaptive two-way relay cooperation scheme is studied for multiple-relay cognitive radio networks to improve the performance of secondary transmissions. The power allocation and relay selection schemes are derived to minimize the secondary outage probability where only statistical channel information is needed. Exact closed-form expressions for secondary outage probability are derived under a constraint on the quality of service of primary transmissions in terms of the required primary outage probability. To better understand the impact of primary user interference on secondary transmissions, we further investigate the asymptotic behaviors of the secondary relay network including power allocation and outage probability, when the primary signal-to-noise ratio goes to infinity. Simulation results are provided to illustrate the performance of the proposed schemes.
\end{abstract}

\begin{IEEEkeywords}
Two-way relay, cognitive radio networks, outage probability, power allocation, relay selection.
\end{IEEEkeywords}

\IEEEpeerreviewmaketitle

\section{Introduction}
\IEEEPARstart{C}{ognitive} radio techniques enable secondary users (SUs) to access the frequency bands originally licensed to primary users (PUs) while ensuring that the quality of service (QoS) of primary transmissions is not affected, which can improve spectral efficiency significantly \cite{haykin2005cognitive}. However, the SUs often operate with constrained transmit power to guarantee the QoS of PUs in terms of interference temperature, thus limiting the throughput and coverage of the secondary system. To combat this problem, cooperative diversity systems involving scattering relay networks have recently been researched to exploit the spatial diversity gain and to enhance the secondary channel performance \cite{jing2006distributed}, \cite{4786524}. It has also been shown that cooperative diversity with relay selection can achieve the same diversity-multiplexing tradeoff as achieved by the traditional cooperation protocols where all relays are involved in forwarding the signals from source nodes \cite{hunter2006outage}, \cite{zou2010adaptive}.

The conventional one-way relay scheme suffers from a loss in spectral efficiency because of half-duplex transmission \cite{laneman2004cooperative}. To circumvent this disadvantage, a two-way relay system was proposed in \cite{rankov2007spectral}. A two-way relay system has two transmission phases. During the first phase, two secondary transceivers (STs) simultaneously broadcast their signals. After successfully receiving the combined signals, the relay node forwards the signals to the two STs during the second phase. Since there are two different relaying paths, the total spectral efficiency of a two-way relay system can be doubled compared with a conventional one-way relay system. Two protocols for two-way relay networks, commonly known as decode-and-forward (DF) and amplify-and-forward (AF) relaying, were proposed in \cite{rankov2007spectral}. Based on these, several cooperative diversity schemes for two-way relay networks with relay selection have been proposed \cite{jing2009relay,li2010asymptotically,krikidis2010relay,talwar2011joint}. Note that all the aforementioned works studied non-cognitive radio networks. However, in practical cognitive radio systems, PUs and SUs can simultaneously transmit signals by sharing the same spectrum resources. As a result, the relays and secondary receivers inevitably suffer interference from PUs. From the viewpoint of SUs, these interferences come in the form of co-channel interference (CCI) and it is important to analyze their effect on system performance. 

\subsection{Related Work}
So far, the literature that studies outage performance and resource allocation in cognitive two-way relaying networks with CCI is relatively scarce. Interference was considered only during the second transmission phase in \cite{6403863}, where exact outage probability was obtained while ignoring the noise at the receivers. In \cite{zhang2015exact}, the exact outage probability was derived under a cognitive two-way relay network setting. However, the system outage event was defined as having either one of the two STs in outage, which simplifies the derivation but does not represent system outage correctly. In \cite{ubaidulla2012optimal}, a max-min strategy over instantaneous achievable channel rates was employed to address relay selection and power allocation for cognitive two-way AF relaying networks. The CCI from the PUs was modeled as Gaussian noise, which does not characterize the practical cognitive radio communication appropriately. Relay selection and power allocation schemes in cognitive two-way DF relaying network were studied for the first time to maximize the achieved sum rate in \cite{alsharoa2013relay}. However, the CCI was considered at primary nodes whereas the interference resulting from primary transmission in secondary receivers was not considered. In \cite{alsharoa2014optimal}, the power allocation problem in the cognitive two-way relay network with amplify-and-forward strategy was studied and the secondary sum rate was maximized whereas the optimization problem dealt with the terminal side without any control on relay parameters. Instantaneous secrecy rate was maximized in \cite{liu2015relay} for relay selection, which is the same as maximizing signal-to-noise ratio (SNR). Besides the inappropriate system modeling of CCI, the resource allocation schemes in the aforementioned works lead to the maximization of the instantaneous SNR. These resulting resource allocation schemes require perfect knowledge of instantaneous channel state information (CSI) between the nodes in the cognitive network. In fact, it is highly computationally complex and also sometimes impossible to accurately learn the knowledge of instantaneous CSI in the network. Moreover, in the cognitive radio network setting, the knowledge of instantaneous CSI for the primary interference transmitted from the primary network to the secondary network is required if those schemes are to be implemented, which is extremely difficult if there are no pilot symbols specifically designed for the secondary nodes in the primary signal. Therefore, optimal power allocation for outage probability minimization comes into consideration in such a scenario, which only requires the knowledge of statistical CSI \cite{ahmed2004outage,dulek2013optimum,huang2014optimal}.

\subsection{Main Contributions}	
In this paper, we investigate an adaptive cooperative diversity scheme in cognitive two-way relay networks using the DF protocol, where mutual interference between PUs and SUs is considered. The STs broadcast their signals to the relays and to each other through the direct link during the first phase. During the second phase, if the relays can decode the signals received during the first phase, the best relay is chosen to forward the signals to the STs; otherwise, the STs adaptively repeat the same transmission to each other through the direct link as during the first phase. Then, the STs combine the two copies of the received signals after the two transmission phases. 

{The main contributions of this paper are as follows:

 1) We explore the adaptive use of the direct link and the relay link to achieve higher system performance in cognitive DF two-way relaying networks. Our analysis can also be used for the scenario where only relay link is available.
 
 2) For the first time, a power allocation scheme for STs and the relays is developed that minimizes the secondary outage probability under a QoS constraint from the primary network, requiring no instantaneous CSI of the transmission links. 
 
 3) The optimal relay selection approach for this two-way system is also provided to minimize outage probability, which requires only statistical CSI. To the best of our knowledge, this is the first work to study the resource allocation problem using statistical CSI information for the proposed general framework. An exact closed-form expression for the secondary outage probability is also derived in this paper. Asymptotic behavior of the secondary system is analyzed given that the primary user SNR goes to infinity.}

The rest of the paper is organized as follows: we give the system model in Section II. Section III provides the outage analysis of the relaying network. Based on the outage analysis, we address the power allocation and relay selection problems in Section IV. We analyze the asymptotic behavior of the system in Section V. Numerical simulation results and conclusion are given in Section VI and Section VII.
\section{Proposed Adaptive Cooperation Scheme}
\subsection{System Model}
Consider a general spectrum-sharing cognitive two-way relaying network as shown in Fig. 1. In the primary network, a primary transmitter $u$ sends data to a primary destination $v$. Meanwhile, in the secondary relay network, STs $s$ and $d$ exchange information with each other. Secondary relays $r_i, i=1,2,3,\dots,M$, are available to assist secondary data transmissions using the DF protocol. We assume that the channel link from  $k$ to $j$ $(k,j \in \{ u,v,s,r_i,d\} )$ undergoes Rayleigh fading with instantaneous coefficient $h_{k,j}$. Therefore, the channel gain ${\left| {{h_{k,j}}} \right|^2}$ is exponentially distributed with mean $\frac{1}{\sigma_{k,j}^2}$. We also assume reciprocity of all the channels and zero-mean additive white Gaussian noise (AWGN) with variance $N_0$ at each receiver.
\begin{figure}[!t]
\centering
\includegraphics[width=2.5in]{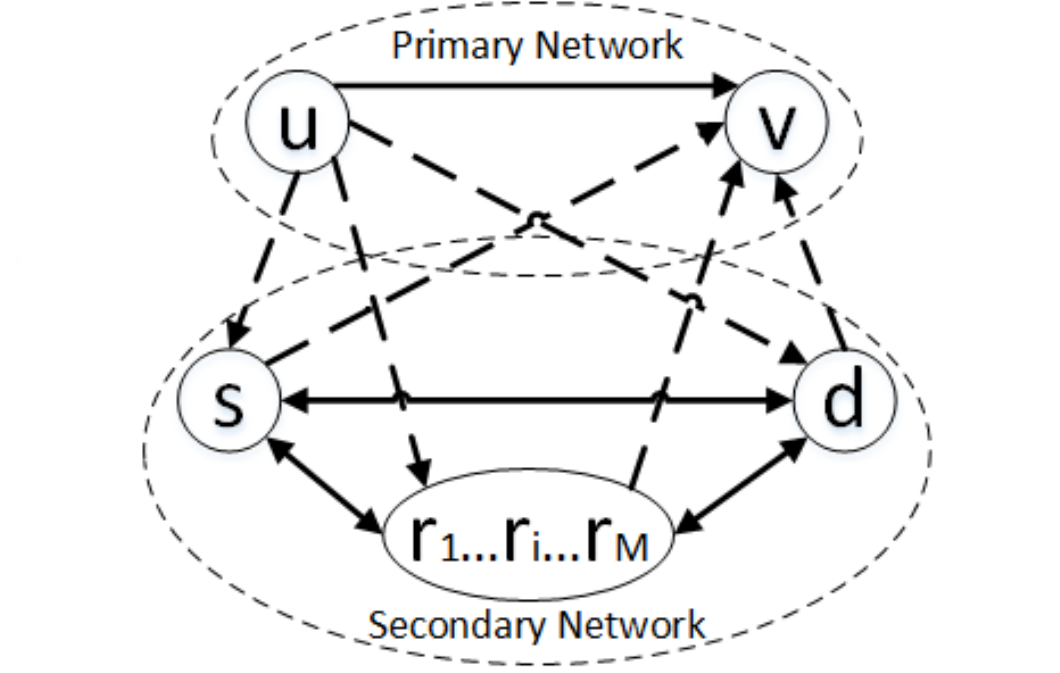}
\caption{\small System model of a cognitive two-way relaying network. The lines between the nodes denote data transmission or interference. Solid lines represent data transmission that takes place either during the first or the second phase. The dashed lines stand for co-channel interference between primary and secondary nodes.}
\label{f1}
\end{figure}

During the first phase, STs $s$ and $d$ simultaneously broadcast their signals to the relay $r_i$ and to the corresponding receiver, i.e., $s \to {r_i} \leftarrow d$, $s \leftrightarrow d$. By employing multiple antennas and self-interference cancelation (SIC), the STs can send and receive at the same time \cite{ju2009catching}. Thus, considering coexistence of primary transmission, the received signal at the primary receiver can be expressed as
\begin{align}
{y_v} = \sqrt {{{\mathbb P}_u}} {h_{u,v}}{x_u} + \sqrt {{{\mathbb P}_s}} {h_{s,v}}{x_s} + \sqrt {{{\mathbb P}_d}} {h_{d,v}}{x_d} + {n_v},
\end{align}
where ${{\mathbb P}_u}$, ${{\mathbb P}_s}$ and ${{\mathbb P}_d}$ are the transmit powers of $u$, $s$ and $d$ respectively, ${x_u}$, ${x_s}$ and ${x_d}$ denote the unit-mean-energy symbols transmitted respectively by $u$, $s$ and $d$, and ${n_v}$ is the AWGN. The QoS of primary transmissions is quantified by the outage probability in this paper. The primary QoS guarantee is represented by the inequality that the outage probability of primary transmission $P_{uv}$ does not exceed a predefined outage probability threshold $P_{th}$, which is expressed as
\begin{align}\label{1}
{P_{uv}} = &   P\left( \log _2 \left(1+{\frac{{{{\mathbb P}_u}{{\left| {{h_{u,v}}} \right|}^2}}}{{{{\mathbb P}_s}{{\left| {{h_{s,v}}} \right|}^2} + {{\mathbb P}_d}{{\left| {{h_{d,v}}} \right|}^2} + {N_0}}} } \right ) < R_u\right)  \nonumber \\
=&P\left( {\frac{{{{\mathbb P}_u}{{\left| {{h_{u,v}}} \right|}^2}}}{{{{\mathbb P}_s}{{\left| {{h_{s,v}}} \right|}^2} + {{\mathbb P}_d}{{\left| {{h_{d,v}}} \right|}^2} + {N_0}}} < \Delta_u} \right) \le {P_{th}},
\end{align}
where ${\Delta _u} = {2^{{R_u}}} - 1$ with $R_u$ being the primary data rate. We calculate $P_{uv}$ and write the primary QoS guarantee during the first phase as
\begin{align}\label{15}
{P_{uv}} = 1 -& \int_0^\infty  \int_0^\infty  \exp \left[ { - \frac{{{\Delta _u}}}{{{\gamma _u}\sigma _{u,v}^2}}\left( {{\gamma _s}x + {\gamma _d}y + 1} \right)} \right]\nonumber\\
&\times{f_{{{\left| {{h_{s,v}}} \right|}^2}}}\left( x \right){f_{{{\left| {{h_{d,v}}} \right|}^2}}}\left( y \right)dxdy \nonumber\\
 = 1 - &\frac{{{{\left( {{\gamma _u}\sigma _{u,v}^2} \right)}^2}\exp \left( { - \frac{{{\Delta _u}}}{{{\gamma _u}\sigma _{u,v}^2}}} \right)}}{{\left( {{\gamma _u}\sigma _{u,v}^2 + {\Delta _u}{\gamma _s}\sigma _{s,v}^2} \right)\left( {{\gamma _u}\sigma _{u,v}^2 + {\Delta _u}{\gamma _d}\sigma _{d,v}^2} \right)}} \le {P_{th}},
\end{align}
where ${\gamma _u} = {{\mathbb P}_u}/{N_0}$, ${\gamma _s} = {{\mathbb P}_s}/{N_0}$ and ${\gamma _d} = {{\mathbb P}_d}/{N_0}$ denote the transmit SNR at the primary transmitter, and STs $s$ and $d$, respectively.
In order to ensure \eqref{15}, we adopt the static power allocation scheme to guarantee the QoS for primary transmission. First we rewrite \eqref{15} and find the constraint in terms of the power of the STs $s$ and $d$ which is denoted by $\mathbf C$ where
\begin{align} 
{\mathbf C}:\ \ \left( {1 + \frac{{{\Delta _u}{{\mathbb P}_s}\sigma _{s,v}^2}}{{{{\mathbb P}_u}\sigma _{u,v}^2}}} \right)\left( {1 + \frac{{{\Delta _u}{{\mathbb P}_d}\sigma _{d,v}^2}}{{{{\mathbb P}_u}\sigma _{u,v}^2}}} \right) \le g,\end{align}
and $g = \max \{ {{\exp \left( { - \frac{{{\Delta _u}}}{{{\gamma _u}\sigma _{u,v}^2}}} \right)}}/{{(1 - {P_{th}})}},1\} $.

\subsection{Proposed Adaptive Cooperation with Relay Selection}
In this subsection, we focus on the adaptive relay cooperation scheme with relay selection. During the first phase, some relays may successfully decode the received signals, among which the best relay is chosen to forward the data to STs. First of all, the received signal during the first phase at $r_i$ is represented as
\begin{align}
{y_{{r_i}}} = \sqrt {{{\mathbb P}_u}} {h_{u,{r_i}}}{x_u} + \sqrt {{{\mathbb P}_s}} {h_{s,{r_i}}}{x_s} + \sqrt {{{\mathbb P}_d}} {h_{d,{r_i}}}{x_d} + {n_{{r_i}}},
\end{align}
where ${n_{{r_i}}}$ is the zero-mean AWGN with variance $N_0$. For convenience, we denote $D_M=\left\{r_1,r_2,\dots,r_M\right\}$ to be the set of all the relays and those relays that are able to successfully decode the received signals constitute a set $D$. { Therefore, $D$ is a dynamic relay set that depends on the decoding status of the relays. Note that the relay set $D$ determines whether a direct transmission between STs is needed or not. If it is not needed, the relay set $D$ determines the feasible relay that can be chosen to forward the signal. }  

During the second phase, if $D$ is empty, i.e., $D = \varnothing$, STs $s$ and $d$ will repeat the transmission of the original signals to each other through the direct link. In this case, with SIC and signal combination using maximum ratio combining (MRC) method, the SINR at each ST can be respectively expressed as
\begin{align}
SIN{R_s}(D = \varnothing) = \frac{{2{{\mathbb P}_d}{{\left| {{h_{d,s}}} \right|}^2}}}{{{{\mathbb P}_u}{{\left| {{h_{u,s}}} \right|}^2} + {N_0}}},\\
SIN{R_d}(D = \varnothing) = \frac{{2{{\mathbb P}_s}{{\left| {{h_{s,d}}} \right|}^2}}}{{{{\mathbb P}_u}{{\left| {{h_{u,d}}} \right|}^2} + {N_0}}}.
\end{align}
Otherwise, if $D$ is not empty, where $D = {D_S}$, the relay $r_{  i}$ chosen within $D_S$ will transmit its decoded data stream to the two STs. Finally, STs combine the two copies of the received signals using SIC and MRC methods. Therefore, the respective SINR is given as
\begin{align}
\label{16}
SIN{R_s} = \frac{{{{\mathbb P}_d}{{\left| {{h_{d,s}}} \right|}^2}}}{{{{\mathbb P}_u}{{\left| {{h_{u,s}}} \right|}^2} + {N_0}}} + \frac{{\beta_{  i} {{\mathbb P}_{{r_{  i}}}}{{\left| {{h_{{r_{  i}},s}}} \right|}^2}}}{{{{\mathbb P}_u}{{\left| {{h_{u,s}}} \right|}^2} + {N_0}}},\\
\label{17}
SIN{R_d} = \frac{{{{\mathbb P}_s}{{\left| {{h_{s,d}}} \right|}^2}}}{{{{\mathbb P}_u}{{\left| {{h_{u,d}}} \right|}^2} + {N_0}}} + \frac{{\alpha_{  i} {{\mathbb P}_{{r_{  i}}}}{{\left| {{h_{{r_{  i}},d}}} \right|}^2}}}{{{{\mathbb P}_u}{{\left| {{h_{u,d}}} \right|}^2} + {N_0}}}.
\end{align}
where ${{\mathbb P}_{{r_{  i}}}}$ is the transmit power of $r_{  i}$, and $\alpha_{  i}$ and $\beta_{  i}$ are the ratios of total transmit power at $r_{  i}$ for the transmission of original signals from $s$ and $d$ to $d$ and $s$, respectively.

\section{Outage Performance Analysis}
In this section, we give the analysis of the outage probability of the proposed adaptive relay cooperation scheme. The exact results of secondary outage probability are derived. Based on the results, we shall provide the resource allocation schemes.

We first study the outage in the relay nodes as defined in \eqref{2}.
According to the achievable rate region as discussed in \cite{kim2011achievable,6403863}, the event of each $r_i$ failing to decode the received signals and resulting in outage is denoted as $O(r_i)$ and can be expressed as
\begin{align}\label{2}
O(r_i)  =  \left\{ {{{\bar \gamma }_{s,{r_i}}}   +   {{\bar \gamma }_{d,{r_i}}}   <   \Delta \ \textrm{or}\ {{\bar \gamma }_{s,{r_i}}}   <  {\Delta _s}\ \textrm{or}\ {{\bar \gamma }_{d,{r_i}}}   <   {\Delta _d}} \right\},
\end{align}
where $\Delta=2^{2(R_s+R_d)}-1$, $\Delta_s=2^{2R_s}-1$, $\Delta_d =2^{2R_d}-1$ with $R_s$ and $R_d$ being the data rates at STs $s$ and $d$, respectively, ${{\bar \gamma }_{s,{r_i}}} = {\gamma _s}{\left| {{h_{s,{r_i}}}} \right|^2}/\left( {{\gamma _u}{{\left| {{h_{u,{r_i}}}} \right|}^2} + 1} \right)$ and ${{\bar \gamma }_{d,{r_i}}} = {\gamma _d}{\left| {{h_{d,{r_i}}}} \right|^2}/\left( {{\gamma _u}{{\left| {{h_{u,{r_i}}}} \right|}^2} + 1} \right)$ are correlated and represent the signal-to-interference-plus-noise ratio (SINR) at $r_i$ with respect to signals from $s$ and $d$ respectively.
\begin{proposition}\label{pro1}
	The outage probability of each relay $r_i$ is given as
	\begin{align}\label{5}
	&P(O(r_i)) \nonumber\\
	& =  \left\{ {\begin{array}{*{20}{c}}
		{    1  -  \frac{{T\exp \left( { - \frac{\Delta }{{{\gamma _d}\sigma _{d,{r_i}}^2}}} \right)}}{{{\gamma _u}\sigma _{u,{r_i}}^2}}  \left[ { 1  +  \frac{{{\Delta _s}{\Delta _d}\left( {1 + T} \right)}}{{{\gamma _d}\sigma _{d,{r_i}}^2}}}  \right] ,{\textrm{if}}\;\frac{{{\gamma _d}}}{{{\gamma _s}}}  =  \frac{{\sigma _{s,{r_i}}^2}}{{\sigma _{d,{r_i}}^2}}}\\
		{1 - \frac{{C\exp \left( { - A} \right)}}{{A{\gamma _u}\sigma _{u,{r_i}}^2 + 1}} - \frac{{\left( {1 - C} \right)\exp \left( { - B} \right)}}{{B{\gamma _u}\sigma _{u,{r_i}}^2 + 1}},\textrm{otherwise}}
		\end{array}} \right.,
	\end{align}
	where $T = ({{\frac{\Delta }{{{\gamma _d}\sigma _{d,{r_i}}^2}} + \frac{1}{{{\gamma _u}\sigma _{u,{r_i}}^2}}}})^{-1}$, $A  =  \frac{{\Delta   -  {\Delta _d}}}{{{\gamma _s}\sigma _{s,{r_i}}^2}}  +  \frac{{{\Delta _d}}}{{{\gamma _d}\sigma _{d,{r_i}}^2}}$, $B  =  \frac{{{\Delta _s}}}{{{\gamma _s}\sigma _{s,{r_i}}^2}}  +  \frac{{\Delta   -  {\Delta _s}}}{{{\gamma _d}\sigma _{d,{r_i}}^2}}$ and $C  =   \frac{{{\gamma _s}\sigma _{s,{r_i}}^2}}{{{\gamma _s}\sigma _{s,{r_i}}^2  - {\gamma _d}\sigma _{d,{r_i}}^2}}$.
\end{proposition}
\begin{proof}
	See Appendix \ref{prop1}.
\end{proof}

	As we can see from the proposition,	the outage probability $P(O(r_i))$ is dependent on the transmit powers of the networks, data rates, and the statistical conditions of the channels linked to the relay $r_i$. Note that only the relays that are not in outage can be chosen to forward the signals to the STs. Depending on the channel coefficient $\frac{\sigma^2_{s,r_i}}{\sigma^2_{d,r_i}}$, the outage probability $P(O(r_i))$ takes different forms of expressions. Thus, further analysis in this paper that is based on $P(O(r_i))$ is conducted in a case-by-case fashion.

Now we study the outage behavior of secondary system under the condition that the relay node $r_i$ is chosen and can successfully forward the signals to the STs $s$ and $d$.
Specifically, if $r_{  i}$ forwards the signals as the relay, the secondary system is in outage if at least one of the two STs can not successfully decode the received signal. Let $O({\rm ST}|r_{  i})$ denote the corresponding outage event.
\begin{align}\label{stout}
P(O({\rm ST}|r_{  i})) = P\left( {SIN{R_s} < {\Delta _d}\ {\rm{or}}\ SIN{R_d} < {\Delta _s}} \right),
\end{align}
where $SNIR_s$ and $SINR_d$ are defined in \eqref{16} and \eqref{17}, respectively.
The STs are in outage if they can not receive and decode the signals as is implied by \eqref{stout}.
\begin{proposition}
	In the high SNR regime, i.e., when $N_0 \to 0$, the probability that the STs are in outage is given as
	\begin{align}\label{pop}
	P(O({\rm ST}|r_{  i})) =\mathbb{A}+\mathbb{B}-\mathbb{A}\mathbb{B},
	\end{align}
	where
	$
	\mathbb{A}=\frac{1}{{1 + \frac{{{{\mathbb P}_d}\sigma _{d,s}^2}}{{{{\mathbb P}_u}{\Delta _d}\sigma _{u,s}^2}}}}\frac{1}{{1 + \frac{{{\beta _{  i}}{{\mathbb P}_{{r_{  i}}}}\sigma _{{r_{  i}},s}^2}}{{{{\mathbb P}_u}{\Delta _d}\sigma _{u,s}^2}}}},
	$
	and
	$
	\mathbb{B}=\frac{1}{{1 + \frac{{{{\mathbb P}_s}\sigma _{s,d}^2}}{{{{\mathbb P}_u}{\Delta _s}\sigma _{u,d}^2}}}}\frac{1}{{1 + \frac{{{\alpha _{  i}}{{\mathbb P}_{{r_{  i}}}}\sigma _{{r_{  i}},d}^2}}{{{{\mathbb P}_u}{\Delta _s}\sigma _{u,d}^2}}}}.
	$
\end{proposition}
\begin{IEEEproof}
	See Appendix \ref{prop2}.
\end{IEEEproof}

The probability $P(O({\rm ST}|r_{  i}))$ characterizes the outage property of the secondary system when the relay node $r_{i}$ is chosen to forward the signals to the STs. 
From the expressions of $\mathbb A$ and $\mathbb B$, we observe that the choice of relay has an influence on the secondary system outage performance by the following means: the forward power ratios $\alpha_i$ and $\beta_i$, transmit power of the relay $\mathbb P_{r_i}$, and the channel conditions of the links between the relay and the transceivers $\sigma^2_{r_i,s}$ and $\sigma^2_{r_i,d}$.

Here, we provide the exact probability that the secondary system is in outage.

With \eqref{5}, the probability of the case ${D = \varnothing }$ can be simply given as
\begin{align}\label{10}
P\left( {D = \varnothing } \right) = \prod\limits_{i = 1}^M {P(O(r_i))},
\end{align}
and the outage probability of the secondary network given this case is expressed as
\begin{align}\label{11}
&P({\rm{out}}\left| {D  =  \varnothing }  \right.) \nonumber \\
& =  P \left(   {\frac{{2{\gamma _d}{{\left| {{h_{d,s}}} \right|}^s}}}{{{\gamma _u}{{\left| {{h_{u,s}}} \right|}^s}   +  1}}   <   {\Delta _d}\ \textrm{or}\ \frac{{2{\gamma _s}{{\left| {{h_{s,d}}} \right|}^s}}}{{{\gamma _u}{{\left| {{h_{u,d}}} \right|}^s}   +  1}}   <   {\Delta _s}}   \right)\nonumber\\
&= 1  -  \frac{{2{\gamma _d}\sigma _{d,s}^2 \exp \left( {  - \frac{{{\Delta _d}}}{{2{\gamma _d}\sigma _{d,s}^2}}}  \right)}}{{2{\gamma _d}\sigma _{d,s}^2  +  {\Delta _d}{\gamma _u}\sigma _{u,s}^2}}\frac{{2{\gamma _s}\sigma _{s,d}^2\exp \left( { - \frac{{{\Delta _s}}}{{2{\gamma _s}\sigma _{s,d}^2}}} \right)}}{{2{\gamma _s}\sigma _{s,d}^2  +  {\Delta _s}{\gamma _u}\sigma _{u,d}^2}}.
\end{align}
Similarly, the occurrence probability of the case ${D = D_S }$ is
\begin{align}\label{9}
P\left( {D = {D_S}} \right) = \prod\limits_{{r_i} \in {D_S}} {\left[ {1 - P(O(r_i))} \right]} \prod\limits_{{r_i} \in \overline {{D_S}} } {P(O(r_i))},
\end{align}
where $\overline {{D_S}}=D_M-D_S$ is the complementary set to $D_S$. Based on \eqref{6}, we can also derive the conditional secondary outage probability in this case as
\begin{align}\label{7}
&P\left( {{\rm{out}}|D = {D_S}} \right) \nonumber \\
&= \int\limits_{\mathbb X} {\int\limits_{\mathbb Y} {P\left( {{\rm{out}}|D = {D_S},{\mathbb X},{\mathbb Y}} \right)f\left( {\mathbb X} \right)f\left( {\mathbb Y} \right)d{\mathbb X}d{\mathbb Y}} } ,
\end{align}
with
$
P\left( {{\rm{out}}|D = {D_S},{\mathbb X},{\mathbb Y}} \right) = \prod\limits_{{r_i} \in {D_S}} {\left( {1 - \Phi } \right)} 
$
where 
$
\Phi  = [1 - P(\beta_i {\gamma _{{r_i}}}{\left| {{h_{{r_i},s}}} \right|^2} < {\mathbb {X}})][1 - P(\alpha _i{\gamma _{{r_i}}}{\left| {{h_{{r_i},d}}} \right|^2} < {\mathbb {Y}}),
$
and
$
{\mathbb {X}}={\Delta _d}{\gamma _u}{\left| {{h_{u,s}}} \right|^2} + {\Delta _d} - {\gamma _d}{\left| {{h_{d,s}}} \right|^2}, \\
{\mathbb {Y}}={\Delta _s}{\gamma _u}{\left| {{h_{u,d}}} \right|^2}+ {\Delta _s} - {\gamma _s}{\left| {{h_{s,d}}} \right|^2}. 
$
The PDFs $f({\mathbb {X}})$ and $f({\mathbb {Y}})$ are given in Appendix \ref{mulpdf}.
Then, taking account of various integral intervals and binomial expansion, we derive the expression of $P\left( {{\rm{out}}|D = {D_S},  {\mathbb X},  {\mathbb Y}} \right)$ as in \eqref{8}.
\begin{figure*}[ht]
	\normalsize
	\begin{align}\label{8}
	P\left( {{\rm{out}}|D = {D_S},  {\mathbb X},  {\mathbb Y}} \right) = \left\{ {\begin{array}{*{20}{c}}
		0&{  {\mathbb X} < 0,  {\mathbb Y} < 0}\\
		{1 + \sum\limits_{{D_C} \in {D_S}} {{{\left( { - 1} \right)}^E}} \exp \left( {\sum\limits_{{r_i} \in {D_C}} {\frac{{ -   {\mathbb Y}}}{{{\alpha _i}{\gamma _{{r_i}}}\sigma _{{r_i},d}^2}}} } \right)}&{  {\mathbb X} < 0,  {\mathbb Y} > 0}\\
		{1 + \sum\limits_{{D_C} \in {D_S}} {{{\left( { - 1} \right)}^E}} \exp \left( {\sum\limits_{{r_i} \in {D_C}} {\frac{{ -   {\mathbb X}}}{{{\beta _i}{\gamma _{{r_i}}}\sigma _{{r_i},s}^2}}} } \right)}&{  {\mathbb X} > 0,  {\mathbb Y} < 0}\\
		{1 + \sum\limits_{{D_C} \in {D_S}} {{{\left( { - 1} \right)}^E}} \exp \left[ {\sum\limits_{{r_i} \in {D_C}} {\left( {\frac{{ -   {\mathbb X}}}{{{\beta _i}{\gamma _{{r_i}}}\sigma _{{r_i},s}^2}} - \frac{  {\mathbb Y}}{{{\alpha _i}{\gamma _{{r_i}}}\sigma _{{r_i},d}^2}}} \right)} } \right]}&{  {\mathbb X} > 0,  {\mathbb Y} > 0}
		\end{array}} \right.,
	\end{align}
	where $D_C$ is the non-empty subset of $D_S$ with $E$ elements. 
	
	\hrulefill \vspace*{4pt}
\end{figure*} 

Substituting \eqref{8} into \eqref{7}, we have
\begin{align}\label{12}
&P \left( {{\rm{out}} \left| {D  =  {D_s}} \right.}  \right)  =  1  \nonumber \\
&+  \frac{{\sum\limits_{{D_C} \in {D_S}}  {{\left( { - 1} \right)}^E}\left( {\Omega \Xi  + \Omega \Psi  + \Xi \Lambda  - \Lambda \Psi } \right)}}{{\left(  {{\Delta _d}{\gamma _u}\sigma _{u,s}^2  +  {\gamma _d}\sigma _{d,s}^2}  \right)  \left(  {{\Delta _s}{\gamma _u}\sigma _{u,d}^2  +  {\gamma _s}\sigma _{s,d}^2}  \right)}},
\end{align}
where
\begin{align}
\Omega  = &\frac{{\exp \left( { - \frac{{{\Delta _d}}}{{{\gamma _d}\sigma _{d,s}^2}}} \right) - \exp \left( { - \sum\limits_{{r_i} \in {D_C}} {\frac{{{\Delta _d}}}{{\beta_i {\gamma _{{r_i}}}\sigma _{{r_i},s}^2}}} } \right)}}{{\sum\limits_{{r_i} \in {D_C}} {\frac{1}{{\beta_i {\gamma _{{r_i}}}\sigma _{{r_i},s}^2}}}  - \frac{1}{{{\gamma _d}\sigma _{d,s}^2}}}} \nonumber\\
&+ \frac{{\exp \left( { - \sum\limits_{{r_i} \in {D_C}} {\frac{{{\Delta _d}}}{{\beta_i {\gamma _{{r_i}}}\sigma _{{r_i},s}^2}}} } \right)}}{{\sum\limits_{{r_i} \in {D_C}} {\frac{1}{{\beta_i {\gamma _{{r_i}}}\sigma _{{r_i},s}^2}}}  + \frac{1}{{{\Delta _d}{\gamma _u}\sigma _{u,s}^2}}}},
\end{align}
\begin{align}
\Xi  =& \frac{{\exp \left( { - \frac{{{\Delta _s}}}{{{\gamma _s}\sigma _{s,d}^2}}} \right) - \exp \left( { - \sum\limits_{{r_i} \in {D_C}} {\frac{{{\Delta _s}}}{{\alpha_i {\gamma _{{r_i}}}\sigma _{{r_i},d}^2}}} } \right)}}{{\sum\limits_{{r_i} \in {D_C}} {\frac{1}{{\alpha_i {\gamma _{{r_i}}}\sigma _{{r_i},d}^2}}}  - \frac{1}{{{\gamma _s}\sigma _{s,d}^2}}}} \nonumber\\
&+ \frac{{\exp \left( { - \sum\limits_{{r_i} \in {D_C}} {\frac{{{\Delta _s}}}{{\alpha_i {\gamma _{{r_i}}}\sigma _{{r_i},d}^2}}} } \right)}}{{\sum\limits_{{r_i} \in {D_C}} {\frac{1}{{\alpha_i {\gamma _{{r_i}}}\sigma _{{r_i},d}^2}}}  + \frac{1}{{{\Delta _s}{\gamma _u}\sigma _{u,d}^2}}}},
\end{align}
\begin{align}
\Lambda  =& {\gamma _d}\sigma _{d,s}^2\exp \left( {  \frac{{{-\Delta _d}}}{{{\gamma _d}\sigma _{d,s}^2}}} \right),\Psi  = {\gamma _s}\sigma _{s,d}^2\exp \left( {  \frac{{{-\Delta _s}}}{{{\gamma _s}\sigma _{s,d}^2}}} \right).
\end{align}
Finally, we derive the outage probability of the secondary two-way relay network as
\begin{align}
{P_{{\rm{out}}}}    =    & P  \left(  {{\rm{out}}|D   =   \varnothing }  \right)  P  \left(  {D   =   \varnothing }  \right)    \nonumber \\
& +                \sum\limits_{{D_S} \in {D_M}}               {P  \left(  {{\rm{out}}|D   =   {D_S}}  \right)  P  \left( {D   =   {D_S}} \right)},
\end{align}
where $P \left(  {D  =  \varnothing }  \right)$, $P \left(  {{\rm{out}}|D  =  \varnothing }  \right)$, $P  \left( {D  =  {D_S}} \right)$ and $P \left(  {{\rm{out}}|D  =  {D_S}}  \right)$ are given in \eqref{10}, \eqref{11}, \eqref{9} and \eqref{12}, respectively.

\section{Power Allocation and Relay Selection}

In the following, we optimize the outage performance of the secondary receivers in the relay network to address the problems of power allocation and relay selection.
In the context of power allocation for the DF relaying network, we have to determine the powers of STs $s$ and $d$, represented by ${\mathbb P}_s$ and ${\mathbb P}_d$, power of the relay $r_i$, represented by ${\mathbb P}_{r_i}$, and the power ratio for transmission of different data streams at the relay, represented by $\alpha_i$ and $\beta_i$. 

{Next, we determine the power allocation for the STs, i.e., ${\mathbb P}_s$ and ${\mathbb P}_d$.
First, note that the quality of the direct link between the STs may be severely affected due to long distance. This also partially constitutes the reason to employ relays since the links between the STs and the relays are relatively of much higher quality as well as providing diversity. To effectively make use of the relay channel diversity to enhance system performance, we maximize the minimum probability that the link between the STs and a relay is connected, while ${\mathbb P}_d$ and ${\mathbb P}_s$ satisfy constraint $\bf C$, which can be expressed as}
\begin{align}\label{opt}
\left\{ {{{\mathbb P}_s},{{\mathbb P}_d}} \right\} &= \arg \mathop {\max }\limits_{{\rm{subject\ to}}\;{\bf{C}}} \left\{ {\min \left\{ {1 - P(O(r_i)) } \right\}} \right\}\nonumber\\
& = \arg \mathop {\max }\limits_{{\rm{subject\ to}}\;{\bf{C}}} \left\{ {1-P(O(r_{\min}))} \right\}\nonumber\\
& = \arg \mathop {\min }\limits_{{\rm{subject\ to}}\;{\bf{C}}}  {P(O(r_{\min}))},
\end{align}
where
$
{r_{\min }} = \arg \mathop {\min }\limits_{{r_i} \in {D_M}} \left\{ {1 - P(O(r_i)) } \right\}.
$

Recalling \eqref{5}, we give the optimal power allocation of $\{{\mathbb P}_s, {\mathbb P}_d\}$ to minimize $P(O(r_{\min})) $ while satisfying the constraint $\mathbf C$.

Let $P\left( {{O(r_{\text{min}})}}, {\mathbb P}_s, {\mathbb P}_d \right)$ represent the corresponding $P(O(r_{\min}))$ with respect to $\{{\mathbb P}_s, {\mathbb P}_d\}$.
We provide the integrated power allocation strategy of $\{{\mathbb P}_s, {\mathbb P}_d\}$ in the following lemma.
\begin{lemma}\label{lemma1}
	The optimal power allocation $\{{\mathbb P}_s, {\mathbb P}_d\}$ to minimize $P(O(r_{\min})) $ is given by
	\begin{align}
	\left\{ {{{\mathbb P}_s},{{\mathbb P}_d}} \right\} = \arg \mathop {\min }\limits_{\left\{ {{{\mathbb P}_s},{{\mathbb P}_d}} \right\} \in \left\{ {\left\{ {{{\mathbb P}'_s},{{\mathbb P}'_d}} \right\},\left\{ {{{\mathbb P}''_s},{{\mathbb P}''_d}} \right\}} \right\}} \left\{ {P\left( {{O(r_{\text{min}})}}, {\mathbb P}_s, {\mathbb P}_d \right)} \right\},
	\end{align}
	where 
$
	{{\mathbb P}'_s} = \sqrt {{{\left( {\frac{{\sigma _{d,{r_{\text{min}}}}^2}}{{2{\mathbf B}\sigma _{s,{r_{\text{min}}}}^2}} - \frac{1}{{2{\mathbf A}}}} \right)}^2}   +   \frac{{g\sigma _{d,{r_{\text{min}}}}^2}}{{{\mathbf A}{\mathbf B}\sigma _{s,{r_{\text{min}}}}^2}}}  - \frac{{\sigma _{d,{r_{\text{min}}}}^2}}{{2{\mathbf B}\sigma _{s,{r_{\text{min}}}}^2}} - \frac{1}{{2{\mathbf A}}},
$
$
	{\mathbb P}'_d={\mathbb P}'_s\frac{\sigma_{s,r_{\text{min}}}^{2}}{\sigma_{d,r_{\text{min}}}^{2}},
$
$
	{{\mathbb P}''_s} = \frac{{g - 1}}{{\sqrt {\frac{{{\Delta _d}\sigma _{s,{r_{\text{min}}}}^2}}{{{\Delta _s}\sigma _{d,{r_{\text{min}}}}^2}}\mathbf A \mathbf B g}  + \mathbf A}},
$
$
	{{\mathbb P}''_d} = \frac{{g - 1}}{{\sqrt {\frac{{{\Delta _s}\sigma _{d,{r_{\text{min}}}}^2}}{{{\Delta _d}\sigma _{s,{r_{\text{min}}}}^2}}\mathbf A \mathbf B g}  + \mathbf B}}.
$
\end{lemma}
with  ${\mathbf A}=\frac{{{\Delta _u}\sigma _{s,v}^2}}{{{{\mathbb P}_u}\sigma _{u,v}^2}}$ and ${\mathbf B}=\frac{{{\Delta _u}\sigma _{d,v}^2}}{{{{\mathbb P}_u}\sigma _{u,v}^2}}$
\begin{proof}
	See Appendix \ref{pa_pspd}.
\end{proof}

Then, let us look at the allocation schemes at the relay, i.e., the power allocation for the relay node and relay selection. We give the power of the relay, i.e., ${\mathbb P}_{r_{  i}}$ in the following lemma. 
\begin{lemma}
	The optimal transmit power of the chosen best relay $P_{r_{  i}}$ is given by
	\begin{align}\label{14}
	{{\mathbb P}_{{r_{  i}}}} = \frac{{{{\mathbb P}_u}\sigma _{u,v}^2}}{{{\Delta_u}\sigma _{{r_{  i}},v}^2}}(g-1).
	\end{align}
\end{lemma}
\begin{proof}
	See Appendix \ref{pa_pri}.
\end{proof}

During the second phase, the selected relay $r_{  i}$ forwards the combined data streams to the STs with power ratios $\alpha_{  i}$ and $\beta_{  i}$. Here, we address the relay selection and provide the power allocation for optimal $\alpha_i$ and $\beta_i$ to minimize overall secondary system outage probability. Note that when $r_i$ is chosen for relaying, the secondary outage probability is $	P(O({\rm ST}|r_{  i}))$, which is hereby to be minimized.

\begin{lemma}\label{lemma3}
	The optimal power ratios of $\alpha_{  i}$ and $\beta_{  i}$ are given by
	\begin{equation}
	{\alpha _{  i}} = \left\{ {\begin{array}{*{20}{c}}
		{\frac{{bd + d - b}}{{2bd}}}&{ab = cd}\\
		{\frac{{ab + a + c - 1}}{{ab - cd}} - \frac{{\sqrt {\left( {ab - d + ad + abd} \right)\left( {bc - b + cd + bcd} \right)} }}{{\sqrt {bd} \left( {ab - cd} \right)}}}&{ab \ne cd}
		\end{array}} \right.
	\end{equation}
	and
	\begin{equation}
	\beta_{  i}=1-\alpha_{  i}
	\end{equation}
	where $a = 1 + \frac{{{{\mathbb P}_d}\sigma _{d,s}^2}}{{{{\mathbb P}_u}{\Delta _d}\sigma _{u,s}^2}}, b = \frac{{{{\mathbb P}_{{r_{  i}}}}\sigma _{{r_{  i}},s}^2}}{{{{\mathbb P}_u}{\Delta _d}\sigma _{u,s}^2}}, c = 1 + \frac{{{{\mathbb P}_s}\sigma _{s,d}^2}}{{{{\mathbb P}_u}{\Delta _s}\sigma _{u,d}^2}}, d = \frac{{{{\mathbb P}_{{r_{  i}}}}\sigma _{{r_{  i}},d}^2}}{{{{\mathbb P}_u}{\Delta _s}\sigma _{u,d}^2}}$.
\end{lemma}
\begin{proof}
	See Appendix \ref{al_be}.
\end{proof}

Therefore, these three lemmas constitute the power allocation scheme including all the transmission powers of the secondary nodes. Note that the derived optimal $\{{\mathbb P}_s, {\mathbb P}_d \}, {{\mathbb P}_{{r_{  i}}}}$ also apply in the case where there is no direct link between the STs.

We substitute the derived optimal $\alpha_i$ and $\beta_i$ back into $	P(O({\rm ST}|r_{  i}))$. {The relay selection scheme selects $r_i \in D_S$ such that the system outage probability $	P(O({\rm ST}|r_{  i}))$, given that $r_i$ is selected, is minimized, which can be written as}
\begin{align}
{r_{  i}} = \mathop {\rm {arg}} \min \limits_{{r_i} \in {D_S}} P(O({\rm ST}|r_{  i})) \label{6}.
\end{align}

It indicates that the proposed relay selection criterion considers the statistical instead of instantaneous CSI of the primary and secondary networks. The benefit of this criterion is prominent since the instantaneous CSI of the networks is typically difficult to obtain. However, the statistical CSI of the primary and secondary networks are much easier for the relays to obtain. Thus, in the particular settings of cognitive relay networks, it is highly desired that power allocation and relay selection request only statistical channel conditions. Note that both centralized and distributed relay selection algorithms can be developed using the proposed relay selection criterion. Specifically, for a centralized relay selection algorithm, an additional node is needed to maintain a table, which consists of $M$ relays and the corresponding statistical CSI. The selection and the related management are completed within this node. For a distributed relay selection algorithm, each relay $r_i$ maintains a timer which is assigned an initial value inversely proportional to min. Therefore, the best relay exhausts its timer the earliest compared with the other relays, and then broadcasts a control packet to notify other relays \cite{zou2013optimal}.

\section{Asymptotic Behavior Analysis}
In a cognitive radio setting, the primary users are licensed to access the channel with QoS guarantee, and the power of primary transmitter is rather high, comparing to secondary transmit power and interference.
In order to have a better understanding of the impact of primary interference on secondary network performance, we analyze the asymptotic behaviors of the derived power allocation and the secondary outage probability when the primary SNR $\gamma_u$ approaches infinity.

First, let us look at the asymptotic behavior of the power allocation. To make it compact and consistent, the power allocation scheme is expressed with respect to SNRs $\{\gamma_s, \gamma_d, \gamma_{r_i} \}$ as well.

When $\gamma_u \rightarrow \infty$, we have
$
g = \max \left\{ {\frac{1}{{1 - {P_{th}}}},1} \right\} = \frac{1}{{1 - {P_{th}}}}\triangleq g'.
$
Let $P\left( {{O(r_{\text{min}})}}, {\gamma}_s, {\gamma}_d \right)$ represent the corresponding asymptotic $P(O(r_{\min}))$ with respect to the  $\{{\gamma}_s, {\gamma}_d\}$.
We provide the integrated asymptotic power allocation strategy of $\{{\gamma}_s, {\gamma}_d\}$ in the following corollary.
\begin{corollary}\label{corollary1}
	The optimal power allocation $\{\gamma_s, \gamma_d\}$ when $\gamma_u \rightarrow \infty$ is given by
	\begin{align}
	\left\{ {{{\gamma}_s},{{\gamma}_d}} \right\} &= \arg \mathop {\min }\limits_{\left\{ {{{\gamma}_s},{{\gamma}_d}} \right\} \in \left\{ {\left\{ {{{\gamma}'_s},{{\gamma}'_d}} \right\},\left\{ {{{\gamma}''_s},{{\gamma}''_d}} \right\}} \right\}} \left\{ {P\left( {{O(r_{\text{min}})}}, {\gamma}_s, {\gamma}_d \right)} \right\}\nonumber\\
	&\triangleq \{ \rho_s\gamma_u,\rho_d\gamma_u\},
	\end{align}
	with $\{\rho_s,\rho_d \}\in \{\{ \rho'_s,\rho'_d\}, \{\rho''_s,\rho''_d \}\}$,
	where
$
{{\gamma '}_s}=\frac{{{\gamma _u}\sigma _{u,v}^2g'}}{{2{\Delta _u}\left( {\sigma _{s,v}^2 + \sigma _{d,v}^2\frac{{\sigma _{s,{r_{{\rm{min}}}}}^2}}{{\sigma _{d,{r_{{\rm{min}}}}}^2}}} \right)}} \triangleq {{\rho '}_s}{\gamma _u},
$
$
	{{\gamma '}_d} = {{\gamma '}_s}\frac{{\sigma _{s,{r_{{\rm{min}}}}}^2}}{{\sigma _{d,{r_{{\rm{min}}}}}^2}} \triangleq {{\rho '}_d}{\gamma _u},
$
$
	{{\gamma ''}_s} = \frac{{\left( {g' - 1} \right)\sigma _{u,v}^2{\gamma _u}}}{{{\Delta _u}\sigma _{s,v}^2\sqrt {g'\frac{{{\Delta _d}\sigma _{s,{r_{{\rm{min}}}}}^2\sigma _{d,v}^2}}{{{\Delta _s}\sigma _{d,{r_{{\rm{min}}}}}^2\sigma _{s,v}^2}} + 1} }} \triangleq {{\rho ''}_s}{\gamma _u},
$
and
$
	{{\gamma ''}_d} = \frac{{\left( {g' - 1} \right)\sigma _{u,v}^2{\gamma _u}}}{{{\Delta _u}\sigma _{d,v}^2\sqrt {g'\frac{{{\Delta _s}\sigma _{d,{r_{{\rm{min}}}}}^2\sigma _{s,v}^2}}{{{\Delta _d}\sigma _{s,{r_{{\rm{min}}}}}^2\sigma _{d,v}^2}} + 1} }} \triangleq {{\rho ''}_d}{\gamma _u}.
$
	
\end{corollary}
\begin{proof}
See Appendix \ref{as_pspd}.
\end{proof}

\begin{corollary}\label{corollary2}
	The optimal power allocation for the relay $\gamma_{r_{i}}$ when $\gamma_u \rightarrow \infty$ is given by
	$
	{\gamma _{{r_i}}} = \frac{{{\gamma _u}\sigma _{u,v}^2}}{{{\Delta _u}\sigma _{{r_{  i}},v}^2}}(g' - 1) \triangleq {\rho _{{r_i}}}{\gamma _u}.
	$
\end{corollary}
\begin{proof}
	It is straightforward and follows from \eqref{14}.
\end{proof}

\begin{corollary}
	The optimal power ratios of $\alpha_{ i}$ and $\beta_{ i}$ when $\gamma_u \rightarrow \infty$ are given by
	\begin{equation}
	{\alpha _{ i}} = \left\{ {\begin{array}{*{20}{c}}
		{\frac{{bd + d - b}}{{2bd}}}&{ab = cd}\\
		{\frac{{ab + a + c - 1}}{{ab - cd}} - \frac{{\sqrt {\left( {ab - d + ad + abd} \right)\left( {bc - b + cd + bcd} \right)} }}{{\sqrt {bd} \left( {ab - cd} \right)}}}&{ab \ne cd}
		\end{array}} \right.
	\end{equation}
	and
	$
	\beta_{ i}=1-\alpha_{ i}
	$,
	where $a = 1 + \frac{{{{\rho}_d}\sigma _{d,s}^2}}{{{\Delta _d}\sigma _{u,s}^2}}, b = \frac{{{{\rho}_{{r_{ i}}}}\sigma _{{r_{ i}},s}^2}}{{{\Delta _d}\sigma _{u,s}^2}}, c = 1 + \frac{{{{\rho}_s}\sigma _{s,d}^2}}{{{\Delta _s}\sigma _{u,d}^2}}, d = \frac{{{{\rho}_{{r_{ i}}}}\sigma _{{r_{  i}},d}^2}}{{{\Delta _s}\sigma _{u,d}^2}}$.
\end{corollary}
\begin{proof}
	Substitute Corollary \ref{corollary1} and Corollary \ref{corollary2} into Lemma \ref{lemma3} and we get the desired result.
\end{proof}

Note that the above three corollaries constitute the power allocation scheme when $\gamma_u \rightarrow \infty$. Interestingly, the transmit powers of secondary transmitters are proportional to the primary transmit power, and the power ratios at the relays are only associated with channel coefficients.

 Next, we investigate the asymptotic secondary outage performance and start with $P(O(r_i))$. Follow the result in Proposition \ref{pro1}, we can have the asymptotic expression of $P(O(r_i))$ in the two cases.

{\bf Case 1}: $\frac{{{\gamma _d}}}{{{\gamma _s}}} = \frac{{\sigma _{s,{r_{i}}}^2}}{{\sigma _{d,{r_{i}}}^2}}$.

In this case, we can obtain the following result as $\gamma_u \rightarrow \infty$:
\begin{align}
P(O({r_i}))&=1 - \frac{T}{{{\gamma _u}\sigma _{u,{r_i}}^2}}\left( {1 + \frac{{{\Delta _s}{\Delta _d}T}}{{{\rho _d}{\gamma _u}\sigma _{d,{r_i}}^2}}} \right) \nonumber\\
&= \frac{{{\Delta ^2}\sigma _{u,{r_i}}^4 + \left( {{\Delta _s} + {\Delta _d}} \right)\sigma _{u,{r_i}}^2{\rho _d}\sigma _{d,{r_i}}^2}}{{{{\left( {\Delta \sigma _{u,{r_i}}^2 + {\rho _d}\sigma _{d,{r_i}}^2} \right)}^2}}}.
\end{align}

{\bf Case 2}: $\frac{{{\gamma _d}}}{{{\gamma _s}}} \ne \frac{{\sigma _{s,{r_{i}}}^2}}{{\sigma _{d,{r_{i}}}^2}}$.

In this case, $P(O({r_i}))$ can be expressed when $\gamma_u \rightarrow \infty$ as:
\begin{align}
P(O({r_i}))&=1 - \frac{C}{{A{\gamma _u}\sigma _{u,{r_i}}^2 + 1}} - \frac{{1 - C}}{{B{\gamma _u}\sigma _{u,{r_i}}^2 + 1}} \nonumber\\
&= \frac{1}{{\left( {{\rho _s}\sigma _{s,{r_i}}^2 - {\rho _d}\sigma _{d,{r_i}}^2} \right)}}\left[ {\frac{{A'{\rho _s}\sigma _{s,{r_i}}^2}}{{\left( {A' + 1} \right)}} + \frac{{B'{\rho _d}\sigma _{d,{r_i}}^2}}{{\left( {B' + 1} \right)}}} \right],
\end{align}
where
$
A' = \frac{{\left( {\Delta  - {\Delta _d}} \right)\sigma _{u,{r_i}}^2}}{{{\rho _s}\sigma _{s,{r_i}}^2}} + \frac{{{\Delta _d}\sigma _{u,{r_i}}^2}}{{{\rho _d}\sigma _{d,{r_i}}^2}}
$
, $
B' = \frac{{\left( {\Delta  - {\Delta _s}} \right)\sigma _{u,{r_i}}^2}}{{{\rho _d}\sigma _{d,{r_i}}^2}} + \frac{{{\Delta _s}\sigma _{u,{r_i}}^2}}{{{\rho _s}\sigma _{s,{r_i}}^2}}.
$

Thus, the asymptotic expression of $P(O({r_i}))$ is obtained. Using \eqref{10} and \eqref{9}, the probability $P\left( {D = \varnothing } \right)$ and $P\left( {D = {D_S}} \right)$ can be obtained respectively. By \eqref{11}, one can also get $P({\rm{out}}\left| {D  =  \varnothing }  \right.)$ when $\gamma_u \rightarrow \infty$ as

\begin{align}
&P({\rm{out}}\left| {D  =  \varnothing }  \right.) \nonumber\\
&= \frac{{2{\rho _d}\sigma _{d,s}^2{\Delta _s}\sigma _{u,d}^2 + 2{\rho _s}\sigma _{s,d}^2{\Delta _d}\sigma _{u,s}^2 + {\Delta _d}\sigma _{u,s}^2{\Delta _s}\sigma _{u,d}^2}}{{\left( {2{\rho _d}\sigma _{d,s}^2 + {\Delta _d}\sigma _{u,s}^2} \right)\left( {2{\rho _s}\sigma _{s,d}^2 + {\Delta _s}\sigma _{u,d}^2} \right)}}
\end{align}

Consequently, inspired by \eqref{12}, we give several asymptotic results as
$
\Omega  = \frac{{{\gamma _u}}}{{\sum\limits_{{r_i} \in {D_C}} {\frac{1}{{{\beta _i}{\rho _{{r_i}}}\sigma _{{r_i},s}^2}}}  + \frac{1}{{{\Delta _d}\sigma _{u,s}^2}}}} \triangleq \Omega '{\gamma _u},
$
$
\Xi  = \frac{{{\gamma _u}}}{{\sum\limits_{{r_i} \in {D_C}} {\frac{1}{{{\alpha _i}{\rho _{{r_i}}}\sigma _{{r_i},d}^2}}}  + \frac{1}{{{\Delta _s}\sigma _{u,d}^2}}}} \triangleq \Xi '{\gamma _u},
$
$
\Lambda  = {\rho _d}\sigma _{d,s}^2{\gamma _u} \triangleq \Lambda '{\gamma _u},\Psi  = {\rho _s}\sigma _{s,d}^2{\gamma _u} \triangleq \Psi '{\gamma _u},
$
and therefore, asymptotic $P\left( {{\rm{out}}\left| {D = {D_s}} \right.} \right)$ is provided as
\begin{align}
&P\left( {{\rm{out}}\left| {D = {D_s}} \right.} \right) \nonumber\\
&= 1 + \frac{{\sum\limits_{{D_C} \in {D_S}} {{{\left( { - 1} \right)}^E}} \left( {\Omega '\Xi ' + \Omega '\Psi ' + \Xi '\Lambda ' - \Lambda '\Psi '} \right)}}{{\left( {{\Delta _d}\sigma _{u,s}^2 + {\rho _d}\sigma _{d,s}^2} \right)\left( {{\Delta _s}\sigma _{u,d}^2 + {\rho _s}\sigma _{s,d}^2} \right)}}.
\end{align}
Since the secondary outage probability is written as
\begin{align}
{P_{{\rm{out}}}}    =    & P  \left(  {{\rm{out}}|D   =   \varnothing }  \right)  P  \left(  {D   =   \varnothing }  \right)    \nonumber \\
& +                \sum\limits_{{D_S} \in {D_M}}               {P  \left(  {{\rm{out}}|D   =   {D_S}}  \right)  P  \left( {D   =   {D_S}} \right)},
\end{align}
the asymptotic behavior of ${P_{{\rm{out}}}} $ when $\gamma_u\rightarrow \infty$ is therefore derived.

Note that the asymptotic secondary outage probability when $\gamma_u\rightarrow \infty$ is only associated with statistical channel coefficients whereas it is independent of $\gamma_u$. Hence, if we characterize the exact secondary outage probability in terms of $\gamma_u$, a horizontal performance floor is expected in the high $\gamma_u$ regime. The underlying reason is that when $\gamma_u$ is large, the secondary transmit SNRs can be expressed linearly of $\gamma_u$. Therefore, in the high $\gamma_u$ regime, secondary signal-to-interference-plus-noise ratios are parameters independent of $\gamma_u$.

\section{Simulation Results}
In this section, we provide simulation results to validate the analysis and to show the improvement brought by the proposed cooperative diversity scheme. Referring to the system model in Fig. \ref{f1}, the
simulation setup is: data rate $R_u=0.6$ bits/s/Hz, $R_d=0.3$ bits/s/Hz, $R_s=0.2$ bits/s/Hz, channel coefficients $\sigma_{u,v}^2=\sigma_{s,d}^2=\sigma_{s,r_i}^2=\sigma_{d,r_i}^2=5$ dB, $\sigma_{u,s}^2 =  \sigma_{u,d}^2 =\sigma_{s,v}^2 = \sigma_{d,v}^2 = \sigma_{u,r_i}^2 = \sigma_{r_i,v}^2 = -5$ dB. The case $M=0$ indicates non-cooperation scheme.
\begin{figure}[h]
\centering
\includegraphics[width=2.8in]{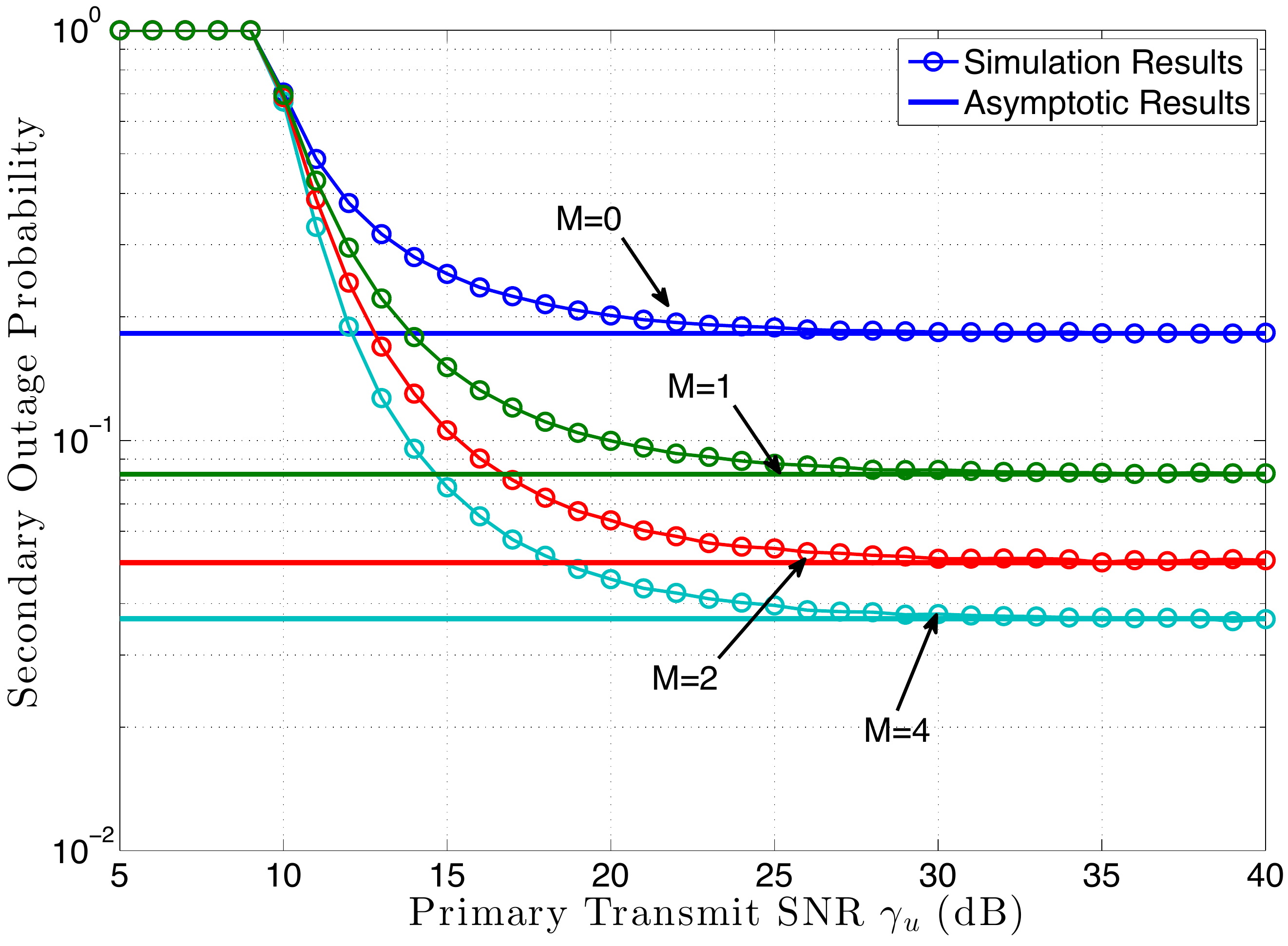}
\caption{\small  Secondary outage probability versus primary transmit SNR $\gamma_u$.}
\label{direct_versus_ga_u}
\end{figure}

First, we set $P_{th}=0.02$ and demonstrate in Fig. \ref{direct_versus_ga_u} the secondary outage probability versus primary SNR $\gamma_u$ of the proposed cooperation and non-cooperation schemes with uniform power allocation, i.e., ${\mathbb P}_s={\mathbb P}_d$ and $\alpha_i=\beta_i$. It is observed that the proposed scheme outperforms the non-cooperation scheme with lower outage probability, which is also improved as the number of relays increases. We notice that the two schemes share the same cutoff value, and secondary transmission is forbidden when $\gamma_u$ is smaller than the cutoff value because no extra interference is allowed in order to achieve the pre-defined primary QoS. A higher $\gamma_u$ results in greater secondary transmit power, and then lower secondary outage probability. As is expected, we can also see a performance floor occurs in high $\gamma_u$ regime, which is due to the fact that the interference from the primary transmitter dominates the secondary outage rather than noise in this case. This also validates the asymptotic outage probability analysis when $\gamma_u \rightarrow \infty$.

\begin{figure}[h]
\centering
\includegraphics[width=2.8in]{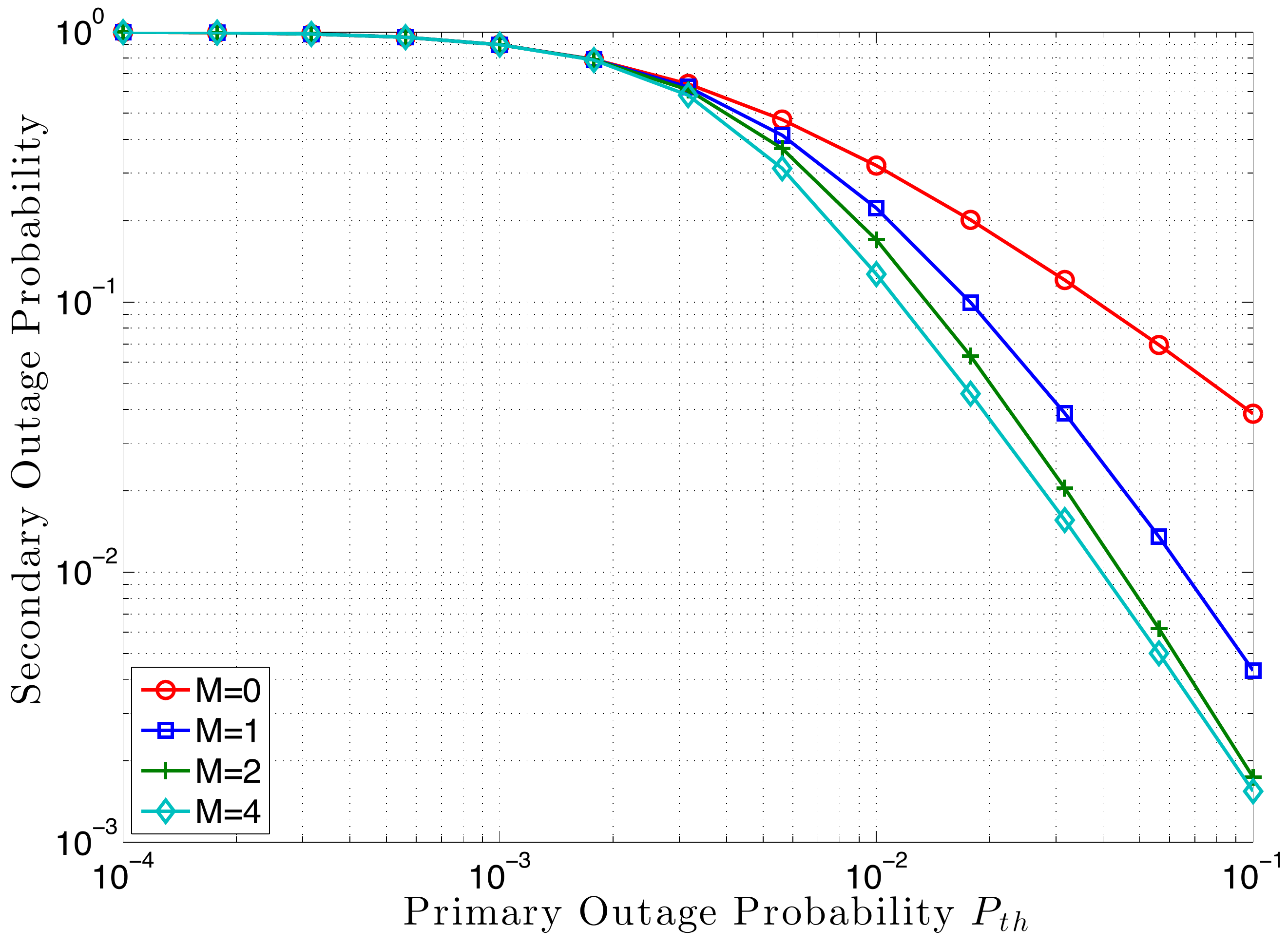}
\caption{\small  Secondary outage probability versus primary QoS constraint $P_{th}$.}
\label{direct_versus_P_th}
\end{figure}

In Fig. \ref{direct_versus_P_th}, we present the secondary outage probability for different values of $P_{th}$. When the QoS requirement of the primary system is too stringent, no secondary transmission is allowed. When the QoS requirement loosens, there begins the secondary transmission and the proposed adaptive cooperation diversity scheme achieves lower outage probability than the non-cooperation scheme. Higher $P_{th}$ allows greater secondary transmit power and then the secondary outage probability is consequently reduced.

Next, we show the simulation results regarding power allocation. The simulation setup is: data rate $R_u=0.6$ bits/s/Hz, $R_d=0.3$ bits/s/Hz, $R_s=0.2$ bits/s/Hz, channel coefficients $\sigma_{u,v}^2=\sigma_{s,r_i}^2=5$ dB, $\sigma_{d,r_i}^2=8$ dB, $\sigma_{s,d}^2=0$ dB, $\sigma_{u,s}^2 =  \sigma_{s,v}^2 = \sigma_{d,v}^2 = \sigma_{u,r_i}^2 = \sigma_{r_i,v}^2 = -5$ dB, $\sigma_{u,d}^2 =-8$ dB.

\begin{figure}[h]
\centering
\includegraphics[width=2.8in]{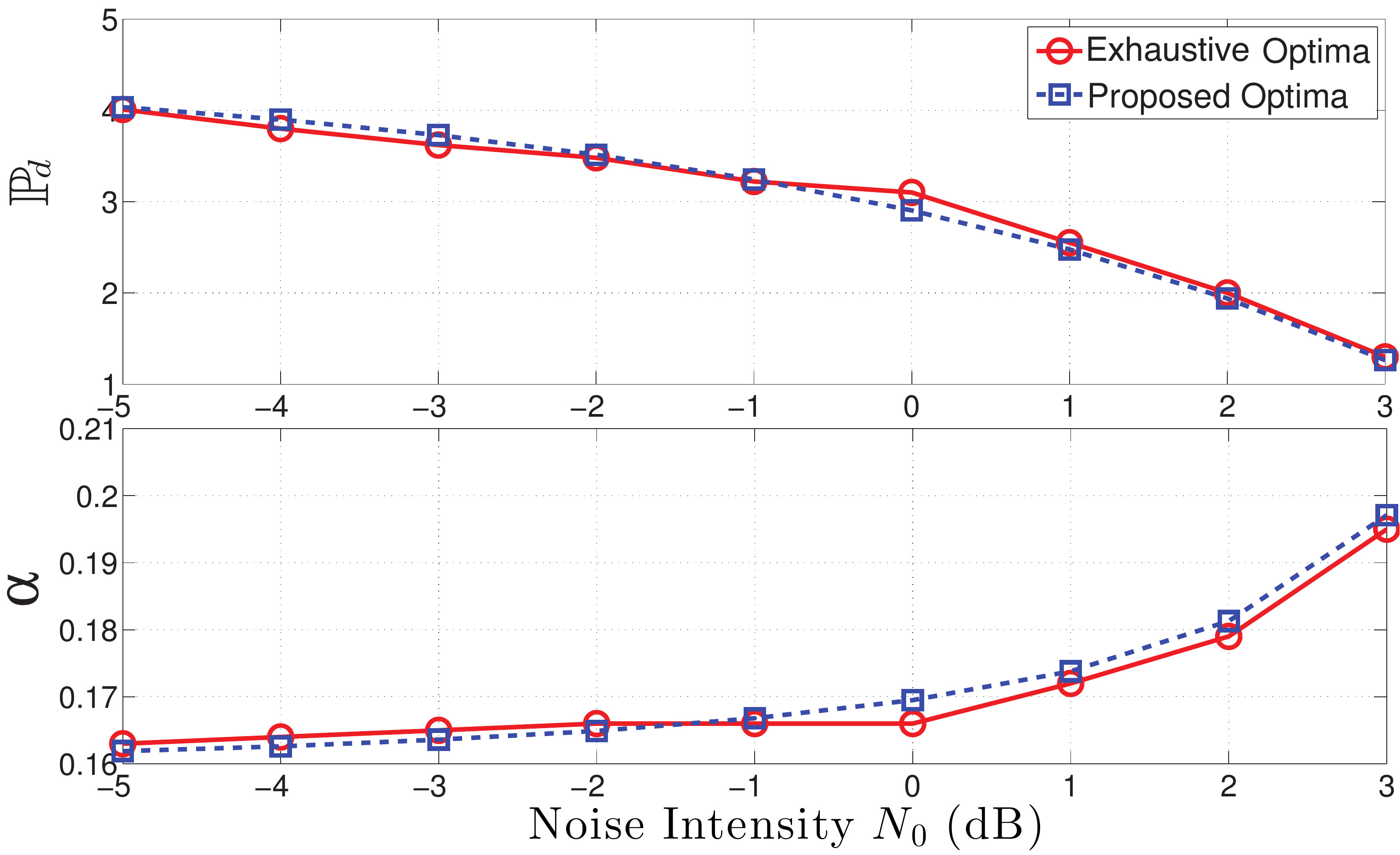}
\caption{\small  Power allocation comparison with various noise intensity $N_0$.}
\label{pa}
\end{figure}

In Fig. \ref{pa}, we plot the power allocation comparison in different noise intensity regimes. Here, the exhaustive optima is obtained by multiple exhaustive search to achieve the minimum overall secondary outage probability. ${\mathbb P}_d$ and $\alpha$ are given in the figure to conduct the comparison. First, we can see the allocated power values given by these two schemes are tightly matched, which indicates the significance of the proposed power allocation scheme. Second, the reason that the scenario $N_0>3$ dB is not given is due to the fact that secondary transmission is switched off to prevent interference to primary users in low SNR regime with the given system parameters. In cognitive relaying networks with high noise intensity, it is highly possible that the secondary transmission is turned down to provide protection to the primary transmission. By this, the high SNR approximation \eqref{app}  in the derivation of optimal $\{{\mathbb P}_s,{\mathbb P}_d \}$ is reasonable.

\begin{figure}[h]
\centering
\includegraphics[width=2.8in]{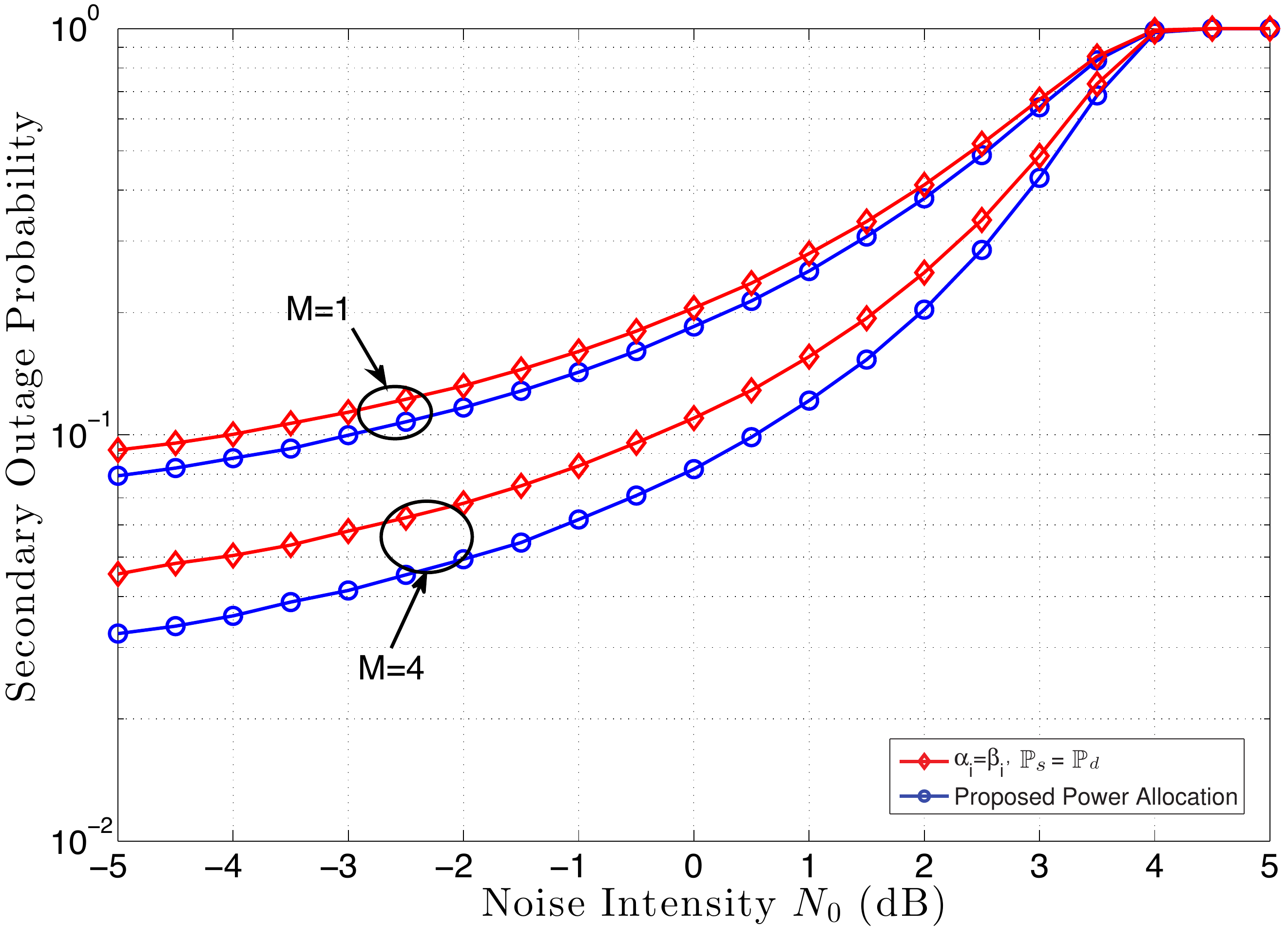}
\caption{\small  Secondary outage probability versus noise intensity $N_0$.}
\label{spa}
\end{figure}

Fig. \ref{spa} shows the secondary outage probability versus noise intensity $N_0$ corresponding to number of relays $M=2$ and $M=4$, respectively. It is observed that the outage performance worsens as noise gets more intense. Also, as the number of relays increases, the outage performance improves. We notice that for all the cases considered, the secondary outage probability approaches 1. This is because in the high $N_0$ scenario, $g = \max \{ {{\exp \left( { - \frac{{{\Delta _u}}}{{{\gamma _u}\sigma _{u,v}^2}}} \right)}}/{{(1 - {P_{th}})}},1\} \rightarrow1$. Then we will find that all the powers of the nodes in the secondary network approach 0. This means that the secondary network is transmitting data with extremely low power in a high noise environment and consequently the outage probability approaches 1. 
We also plot the performance of the secondary network with uniform-power scheme where ${\mathbb P}_s={\mathbb P}_d$ and $\alpha_i=\beta_i$, and this uniform power allocation scheme is widely adopted in two-way relay network literatures. We can see that the proposed power allocation scheme clearly leads to performance improvement compared with the uniform-power scheme, even though the power of relay nodes is both maximized in these two schemes. Another interesting finding is that the proposed power allocation scheme can result in higher relative performance when the secondary network has more relay nodes. Since the proposed scheme is designed to optimize the outage performance and power allocation of ${\mathbb P}_s$ and ${\mathbb P}_d$ considers all the relay channels, more relay nodes enhance the possibility that the given ${\mathbb P}_s$ and ${\mathbb P}_d$ can result in lower outage probability. Although we employ high SNR with asymptotic analysis as a part of our power allocation scheme, the performance improvement can also be seen in the high noise level regime. 

\section{Conclusion}
In this paper, we proposed an adaptive cooperative diversity scheme with power allocation and relay selection in cognitive two-way relay networks. The QoS of the primary network is given by the primary outage probability, which is guaranteed during the transmissions between the secondary users. The closed form of secondary outage probability was derived using the decode-and-forward protocol. To better understand the impact of primary interference on the secondary transmissions, we also investigated the asymptotic behaviors of secondary network when the primary SNR goes to infinity, including power allocation and outage probability. We have presented various simulation results to show the validation of the proposed cooperation scheme. 

The model and the analysis in this paper also suggest potential topics for future research, which include the resource allocation problem in multi-pair cognitive two-way relay networks with imperfect CSI, multiple and paired secondary transceivers, and so on.
\section{Acknowledgement}
We thank Dr. Thakshila Wimalajeewa for helpful discussions and insightful comments on the paper.
\appendices

\section{}\label{prop1}
We rewrite \eqref{2} using Bayes rule as
\begin{align}\label{4}
P(O(r_i)) = \int\limits_z {P'(O(r_i)) \cdot f\left( z \right)dz} ,
\end{align}
where 
$
P'(O(r_i))\triangleq P\left( {O({r_i})\left| { z} \right.} \right),
$
which represents the conditional probability of $P(O(r_i))$ given that ${\gamma _u}{{\left| {{h_{u,{r_i}}}} \right|}^2} + 1 = z$.
Thus, we can further express the conditional probability of $P'(O(r_i))$ given that ${\gamma _s}{\left| {{h_{s,{r_i}}}} \right|^2}=x$ as
\begin{align}
&P'(O(r_i){|{{{\bar \gamma }_s} = x}} )\nonumber\\
&= \left\{ {\begin{array}{*{20}{c}}
	{\ \ \ \ \ \ \ \ \ \ \ \ \ \ \ \ 1,\ \ \ \ \ \ \ \ \ \ \ \ \ \ \ \ \ \ \ \ \ \textrm{if}\ x < {\Delta _s}z}\\
	{P\left( {{{\bar \gamma }_d} < \Delta z - x\ \textrm{or}\ {{\bar \gamma }_d} < {\Delta _d}z} \right),\ \textrm{if}\ x > {\Delta _s}z}
	\end{array}} \right.\nonumber\\
	&= \left\{ {\begin{array}{*{20}{c}}
		{1,\ \ \ \ \ \textrm{if}\ x < {\Delta _s}z}\\
		{1 - \exp \left[ { - \frac{{\left( {\Delta z - x} \right)}}{{{\gamma _d}\sigma _{d,{r_i}}^2}}} \right],\ \ \textrm{if}\ {\Delta _s}z < x < \left( {\Delta  - {\Delta _d}} \right)z}\\
		{1 - \exp \left( { - \frac{{{\Delta _d}z}}{{{\gamma _d}\sigma _{d,{r_i}}^2}}} \right),\ \textrm{if}\ x > \left( {\Delta  - {\Delta _d}} \right)z}
		\end{array}} \right.,
		\end{align}
		with ${{\bar \gamma }_s}={\gamma _s}{\left| {{h_{s,{r_i}}}} \right|^2}$ and ${{\bar \gamma }_d}={\gamma _d}{\left| {{h_{d,{r_i}}}} \right|^2}$. Consequently, $P'(O(r_i))$ can be obtained by the following integration
		\begin{align}\label{3}
		&P'(O(r_i)) = \int\limits_x {P'(O({r_i})|{{\bar \gamma }_s} = x)f(x)dx} \nonumber\\
		&=  \left\{ {\begin{array}{*{20}{c}}
			{    1  -  \left( { 1  +  \frac{{{\Delta _s}{\Delta _d}z}}{{{\gamma _d}\sigma _{d,{r_i}}^2}}}  \right) \exp  \left( {  - \frac{{\Delta z}}{{{\gamma _d}\sigma _{d,{r_i}}^2}}}  \right),{\textrm{if}}\;\frac{{{\gamma _d}}}{{{\gamma _s}}} = \frac{{\sigma _{s,{r_i}}^2}}{{\sigma _{d,{r_i}}^2}}}\\
			{    1  -  C\exp \left( { - Az} \right)  -  \left( {1  -  C} \right)\exp \left( { - Bz} \right),{\textrm{otherwise}}}
			\end{array}} \right..
			\end{align}
			Substituting \eqref{3} into \eqref{4} and with the PDF $f\left( z \right) = \frac{1}{{{\gamma _u}\sigma _{u,{r_i}}^2}}\exp \left[ { - \frac{{\left( {z - 1} \right)}}{{{\gamma _u}\sigma _{u,{r_i}}^2}}} \right]$, we get \eqref{5}.

\section{}\label{prop2}
Let $\mathbb{A}$ and $\mathbb{B}$ respectively represent the probabilities as
\begin{align}
&\mathbb{A}= P\left( {\frac{{{{\mathbb P}_d}{{\left| {{h_{d,s}}} \right|}^2} + {\beta _{  i}}{{\mathbb P}_{{r_{  i}}}}{{\left| {{h_{{r_{  i}},s}}} \right|}^2}}}{{{{\mathbb P}_u}{{\left| {{h_{u,s}}} \right|}^2}}} > {\Delta _d}} \right)\nonumber\\
&= \frac{1}{{1 + \frac{{{{\mathbb P}_d}\sigma _{d,s}^2}}{{{{\mathbb P}_u}{\Delta _d}\sigma _{u,s}^2}}}}\frac{1}{{1 + \frac{{{\beta _{  i}}{{\mathbb P}_{{r_{  i}}}}\sigma _{{r_{  i}},s}^2}}{{{{\mathbb P}_u}{\Delta _d}\sigma _{u,s}^2}}}},
\end{align}
and
\begin{align}
&\mathbb{B}= P\left( {\frac{{{{\mathbb P}_s}{{\left| {{h_{s,d}}} \right|}^2} + {\alpha _{  i}}{{\mathbb P}_{{r_{  i}}}}{{\left| {{h_{{r_{  i}},d}}} \right|}^2}}}{{{{\mathbb P}_u}{{\left| {{h_{u,d}}} \right|}^2}}} > {\Delta _s}} \right)\nonumber\\
&= \frac{1}{{1 + \frac{{{{\mathbb P}_s}\sigma _{s,d}^2}}{{{{\mathbb P}_u}{\Delta _s}\sigma _{u,d}^2}}}}\frac{1}{{1 + \frac{{{\alpha _{  i}}{{\mathbb P}_{{r_{  i}}}}\sigma _{{r_{  i}},d}^2}}{{{{\mathbb P}_u}{\Delta _s}\sigma _{u,d}^2}}}}.
\end{align}
In the high SNR regime, i.e., $N_0 \to 0$, it is straightforward to see that $P(O({\rm ST}|r_{  i})) =\mathbb{A}+\mathbb{B}-\mathbb{A}\mathbb{B}$ from \eqref{16}, \eqref{17}, and \eqref{stout}, and we have the desired result.

\section{}\label{mulpdf}
The PDF of variable ${\mathbb {X}}$ can be expressed as
\begin{align}
&f\left( {\mathbb {X}} \right)  =\frac{{\partial F({\mathbb {X}})}}{{\partial {\mathbb {X}}}}\nonumber\\
&=\frac{{\partial P\left( {{\Delta _d}{\gamma _u}{{\left| {{h_{u,s}}} \right|}^2} + {\Delta _d} - {\gamma _d}{{\left| {{h_{d,s}}} \right|}^2} < {\mathbb {X}}} \right)}}{{\partial {\mathbb {X}}}}\nonumber\\
&=   \frac{{\partial   \int_0^\infty   {P \left(  {{\Delta _d}{\gamma _u}{{\left| {{h_{u,s}}} \right|}^2}   <  {\mathbb {X}}  +  z  -  {\Delta _d} \left| {  z} \right.} \right)f\left( z \right)dz} }}{{\partial {\mathbb {X}}}}\nonumber,
\end{align}
where
$
z={\gamma _d}{{\left| {{h_{d,s}}} \right|}^2} \nonumber
$
and
\begin{align}
&P\left( {{\Delta _d}{\gamma _u}{{\left| {{h_{u,s}}} \right|}^2} < {\mathbb {X}} + z - {\Delta _d}\left| { z} \right.} \right) = \nonumber\\
&\left\{ {\begin{array}{*{20}{c}}
	{0,\ \ \ \ \ \ {\textrm {if}}\ z < {\Delta _d} - {\mathbb {X}}}\\
	{1  -  \exp  \left( { - \frac{{{\mathbb {X}} + z - {\Delta _d}}}{{  {\Delta _d}{\gamma _u}\sigma _{u,s}^2}}} \right),{\textrm {if}}\ \left\{ {\begin{array}{*{20}{c}}
			{z > 0,{\mathbb {X}} > {\Delta _d}}\\
			{z > {\Delta _d} - {\mathbb {X}},{\mathbb {X}} < {\Delta _d}}
			\end{array}} \right.}
	\end{array}} \right.\nonumber.
\end{align}
Thus, the cumulative distribution function is written as
\begin{align}
F({\mathbb {X}})   =   \left\{ {\begin{array}{*{20}{c}}
	{     1  -  \frac{{{\Delta _d}{\gamma _u}\sigma _{u,s}^2}}{{{\Delta _d}{\gamma _u}\sigma _{u,s}^2  + {\gamma _d}\sigma _{d,s}^2}} \exp  \left(  {  \frac{{ {\Delta _d}-{\mathbb {X}}}}{{{\Delta _d}{\gamma _u}\sigma _{u,s}^2}}}  \right) ,{\textrm {if}}\ {\mathbb {X}}  >  {\Delta _d}}\\
	{\frac{{{\gamma _d}\sigma _{d,s}^2}}{{{\Delta _d}{\gamma _u}\sigma _{u,s}^2 + {\gamma _d}\sigma _{d,s}^2}}\exp \left( {\frac{{{\mathbb {X}} - {\Delta _d}}}{{{\gamma _d}\sigma _{d,s}^2}}} \right),{\textrm {if}}\ {\mathbb {X}} < {\Delta _d}}
	\end{array}} \right.\nonumber.
\end{align}
Therefore, we can have the PDF as
\begin{align}
f\left( {\mathbb {X}} \right)  =   \left\{ {\begin{array}{*{20}{c}}
	{\frac{1}{{{\Delta _d}{\gamma _u}\sigma _{u,s}^2 + {\gamma _d}\sigma _{d,s}^2}}\exp \left( {  \frac{{ {\Delta _d}-x}}{{{\Delta _d}{\gamma _u}\sigma _{u,s}^2}}} \right),{\textrm {if}}\ {\mathbb {X}} >  {\Delta _d}}\\
	{\frac{1}{{{\Delta _d}{\gamma _u}\sigma _{u,s}^2 + {\gamma _d}\sigma _{d,s}^2}}\exp \left( {\frac{{x - {\Delta _d}}}{{{\gamma _d}\sigma _{d,s}^2}}} \right),{\textrm {if}}\ {\mathbb {X}} < {\Delta _d}}
	\end{array}} \right.\nonumber.
\end{align}
Similarly, the PDF of variable ${\mathbb {Y}}$ can also be addressed as:
\begin{align}
f\left( \mathbb Y \right)  =   \left\{ {\begin{array}{*{20}{c}}
	{\frac{1}{{{\Delta _s}{\gamma _u}\sigma _{u,d}^2 + {\gamma _s}\sigma _{s,d}^2}}\exp \left( {  \frac{{ {\Delta _s}-\mathbb Y}}{{{\Delta _s}{\gamma _u}\sigma _{u,d}^2}}} \right),{\textrm {if}}\ \mathbb Y >  {\Delta _s}}\\
	{\frac{1}{{{\Delta _s}{\gamma _u}\sigma _{u,d}^2 + {\gamma _s}\sigma _{s,d}^2}}\exp \left( {\frac{{\mathbb Y - {\Delta _s}}}{{{\gamma _s}\sigma _{s,d}^2}}} \right),{\textrm {if}}\ \mathbb Y < {\Delta _s}}
	\end{array}} \right.\nonumber.
\end{align}

\section{}\label{pa_pspd}
To solve the optimization problem in \eqref{opt}, we discuss the solutions in the following two cases.

{\bf Case 1}: $\frac{{{\gamma _d}}}{{{\gamma _s}}} = \frac{{\sigma _{s,{r_{\text{min}}}}^2}}{{\sigma _{d,{r_{\text{min}}}}^2}}$.

In this case, $P(O(r_i))$ decreases monotonically as $\gamma_d$ increases. Substituting $\frac{{{\gamma _d}}}{{{\gamma _s}}} = \frac{{\sigma _{s,{r_{\text{min}}}}^2}}{{\sigma _{d,{r_{\text{min}}}}^2}}$ into primary QoS constraint $\mathbf{C}$ and we have $\{{\mathbb P}_s, {\mathbb P}_d\}$ expressed by $\{{\mathbb P}'_s, {\mathbb P}'_d\}$ as
\begin{align}\label{ps}
{{\mathbb P}'_s} = \sqrt {{{\left( {\frac{{\sigma _{d,{r_{\text{min}}}}^2}}{{2{\mathbf B}\sigma _{s,{r_{\text{min}}}}^2}} - \frac{1}{{2{\mathbf A}}}} \right)}^2}   +   \frac{{g\sigma _{d,{r_{\text{min}}}}^2}}{{{\mathbf A}{\mathbf B}\sigma _{s,{r_{\text{min}}}}^2}}}  - \frac{{\sigma _{d,{r_{\text{min}}}}^2}}{{2{\mathbf B}\sigma _{s,{r_{\text{min}}}}^2}} - \frac{1}{{2{\mathbf A}}},
\end{align}
and
$
{\mathbb P}'_d={\mathbb P}'_s\frac{\sigma_{s,r_{\text{min}}}^{2}}{\sigma_{d,r_{\text{min}}}^{2}},
$
where ${\mathbf A}=\frac{{{\Delta _u}\sigma _{s,v}^2}}{{{{\mathbb P}_u}\sigma _{u,v}^2}}$ and ${\mathbf B}=\frac{{{\Delta _u}\sigma _{d,v}^2}}{{{{\mathbb P}_u}\sigma _{u,v}^2}}$.

{\bf Case 2}: $\frac{{{\gamma _d}}}{{{\gamma _s}}} \ne \frac{{\sigma _{s,{r_{\text{min}}}}^2}}{{\sigma _{d,{r_{\text{min}}}}^2}}$.

By looking at \eqref{5} and constraint $\mathbf{C}$, it is analytically intractable to derive an explicit and closed form expression for optimized ${\mathbb P}_s$ and ${\mathbb P}_d$. Herein, we seek an asymptotic solution in the high SNR regime as $N_0\rightarrow 0$. Recall from \eqref{3} that in this case the conditional probability $P'(O(r_i))$ is given as
\begin{align}
P'(O(r_i))=1 - C\exp ( - Az) - (1 - C)\exp ( - Bz)
\end{align}
where $A = \frac{{{N_0}(\Delta  - {\Delta _d})}}{{{{\mathbb P}_s}\sigma _{s,{r_i}}^2}} + \frac{{{N_0}{\Delta _d}}}{{{{\mathbb P}_d}\sigma _{d,{r_i}}^2}}$, $B = \frac{{{N_0}{\Delta _s}}}{{{{\mathbb P}_s}\sigma _{s,{r_i}}^2}} + \frac{{{N_0}(\Delta  - {\Delta _s})}}{{{{\mathbb P}_d}\sigma _{d,{r_i}}^2}}$, $C = \frac{{{{\mathbb P}_s}\sigma _{s,{r_i}}^2}}{{{{\mathbb P}_s}\sigma _{s,{r_i}}^2 - {{\mathbb P}_d}\sigma _{d,{r_i}}^2}}$, and ${{\left| {{h_{u,{r_i}}}} \right|}^2} + 1 = z$. Therefore, we have the following approximation when $N_0\rightarrow 0$ as
\begin{align}\label{app}
P'(O(r_i))&\approx 1-C(1-Az)-(1-C)(1-Bz)\nonumber\\
&=ACz+B(1-C)z.
\end{align}

Consequently, we can obtain the approximate outage performance in high SNR regime as
\begin{align}
P(O(r_{\min}))  &= \int\limits_z {P'(O(r_{\min})) \cdot f\left( z \right)dz} \nonumber \\
&\approx \int\limits_z {(ACz+B(1-C)z) \cdot f\left( z \right)dz} \nonumber \\
&= \left( {{N_0} + {{\mathbb P}_u}\sigma _{u,{r_{{\rm{min}}}}}^2} \right)\left( {\frac{{{\Delta _s}}}{{{{\mathbb P}_s}\sigma _{s,{r_{{\rm{min}}}}}^2}} + \frac{{{\Delta _d}}}{{{{\mathbb P}_d}\sigma _{d,{r_{{\rm{min}}}}}^2}}} \right).
\end{align}

To find the optimal ${\mathbb P}_s$ and ${\mathbb P}_d$ in this case, we rewrite the constraint $\mathbf C$ in the following form
$
{\mathbf A}{\mathbf B}{\mathbb P}_s{\mathbb P}_d+{\mathbf A}{\mathbb P}_s+{\mathbf B}{\mathbb P}_d+1=g.
$
To find the optimal ${\mathbb P}_s$ and ${\mathbb P}_d$ that minimize $P(O(r_{\min})) $ while satisfying constraint $\mathbf C$, we constitute the Lagrange function as
\begin{align}
{\mathbf L}=\left( {{N_0} + {{\mathbb P}_u}\sigma _{u,{r_{{\rm{min}}}}}^2} \right)\left( {\frac{{{\Delta _s}}}{{{{\mathbb P}_s}\sigma _{s,{r_{{\rm{min}}}}}^2}} + \frac{{{\Delta _d}}}{{{{\mathbb P}_d}\sigma _{d,{r_{{\rm{min}}}}}^2}}} \right) \nonumber\\+ \lambda \left( {\mathbf A}{\mathbf B}{\mathbb P}_s{\mathbb P}_d+{\mathbf A}{\mathbb P}_s+{\mathbf B}{\mathbb P}_d+1-g \right),
\end{align}
with $\lambda$ being the Lagrange multiplier. By solving
\begin{align}
\frac{{\partial \mathbf{L}}}{{\partial {\gamma_s}}} = 0,  \frac{{\partial \mathbf{L}}}{{\partial {\gamma_d}}} = 0,
\end{align}
and considering $\mathbf C$,
we can obtain the power allocation $\{{\mathbb P}_s, {\mathbb P}_d\}$ expressed by $\{{\mathbb P}''_s, {\mathbb P}''_d\}$ in this case as
$
{{\mathbb P}''_s} = \frac{{g - 1}}{{\sqrt {\frac{{{\Delta _d}\sigma _{s,{r_{\text{min}}}}^2}}{{{\Delta _s}\sigma _{d,{r_{\text{min}}}}^2}}\mathbf A \mathbf B g}  + \mathbf A}},
$
and
$
{{\mathbb P}''_d} = \frac{{g - 1}}{{\sqrt {\frac{{{\Delta _s}\sigma _{d,{r_{\text{min}}}}^2}}{{{\Delta _d}\sigma _{s,{r_{\text{min}}}}^2}}\mathbf A \mathbf B g}  + \mathbf B}}.
$

\section{}\label{pa_pri}
During the second transmission phase, the best relay forwards the signals to the STs $s$ and $d$. At the same time, this transmission also causes interference at the primary receivers and the corresponding received signal is expressed as
\begin{align}
{y_v} = \sqrt {{{\mathbb P}_u}} {h_{u,v}}{x_u} + \sqrt {{{\mathbb P}_{{r_{  i}}}}} {h_{{r_{  i}},v}}{x_{{r_{  i}}}} + {n_v}.
\end{align}

Thus, the primary QoS guarantee with respect to outage probability constraint can be written as
\begin{align}
{P_{uv}} = P\left( {\frac{{{{\mathbb P}_u}{{\left| {{h_{u,v}}} \right|}^2}}}{{{{\mathbb P}_{{r_{  i}}}}{{\left| {{h_{{r_{  i}},v}}} \right|}^2} + {N_0}}} < {\Delta _u}} \right) \le {P_{th}},
\end{align}
which can calculated and expressed as
$
{P_{uv}} = 1 - \frac{{\exp \left( { - \frac{{{\Delta _u}{N_0}}}{{{{\mathbb P}_u}\sigma _{u,v}^2}}} \right)}}{{\frac{{{\Delta _u}{{\mathbb P}_{{r_{  i}}}}\sigma _{{r_{  i}},v}^2}}{{{{\mathbb P}_u}\sigma _{u,v}^2}} + 1}} \le {P_{th}}.
$
Therefore, we obtain the transmit power limit of the best relay $r_{  i}$ as
$
{{\mathbb P}_{{r_{  i}}}} \le \frac{{{{\mathbb P}_u}\sigma _{u,v}^2}}{{{\Delta_u}\sigma _{{r_{  i}},v}^2}}(g-1).
$
Note that this is the limit for the transmit power of the best relay and, therefore, we have the optimal ${{\mathbb P}_{{r_{  i}}}}$  when equality is attained, leading to the desired result.

\section{}\label{al_be}
In order to find the optimal values of ${\alpha_{  i}}$ and ${\beta_{  i}} $ that minimize $P(O({\rm ST}|r_{  i}))$ in \eqref{pop}, we construct the Lagrange function as
$
\mathbb{L}= \mathbb{A}+\mathbb{B}-\mathbb{A}\mathbb{B} + \lambda \left( {\alpha_{  i} + \beta_{  i}  - 1} \right),
$
where $\mathbb A$ and $\mathbb B$ are given in \eqref{pop}.

The optimal values of $\alpha_{  i}$ and $\beta_{  i}$ satisfy the equations below
\begin{align}
\frac{{\partial \mathbb{L}}}{{\partial {\alpha_{  i}}}} = 0,  \frac{{\partial \mathbb{L}}}{{\partial {\beta_{  i}}}} = 0, \alpha_{  i} + \beta_{  i}  =1,
\end{align}
which result in the given expressions in the lemma.

\section{}\label{as_pspd}
Following the derivation of Lemma \ref{lemma1}, we begin the asymptotic power allocation of $\{\gamma_s, \gamma_d \}$.

{\bf Case 1}: $\frac{{{\gamma _d}}}{{{\gamma _s}}} = \frac{{\sigma _{s,{r_{\text{min}}}}^2}}{{\sigma _{d,{r_{\text{min}}}}^2}}$.

Let $\{\gamma'_s,\gamma'_d\}$ represent the power allocation scheme in this case. Rewrite \eqref{ps} using the first order Taylor expansion, $\gamma'_s$ is expressed as
\begin{align}
{{\gamma '}_s} &= \frac{g'}{2}\frac{1}{{\frac{{{\Delta _u}\sigma _{s,v}^2}}{{{\gamma _u}\sigma _{u,v}^2}} + \frac{{{\Delta _u}\sigma _{d,v}^2}}{{{\gamma _u}\sigma _{u,v}^2}}\frac{{\sigma _{s,{r_{{\rm{min}}}}}^2}}{{\sigma _{d,{r_{{\rm{min}}}}}^2}}}} \nonumber\\
&= \frac{{{\gamma _u}\sigma _{u,v}^2g'}}{{2{\Delta _u}\left( {\sigma _{s,v}^2 + \sigma _{d,v}^2\frac{{\sigma _{s,{r_{{\rm{min}}}}}^2}}{{\sigma _{d,{r_{{\rm{min}}}}}^2}}} \right)}} \triangleq {{\rho '}_s}{\gamma _u}.
\end{align}
Accordingly, $\gamma'_d$ can also be given
$
{{\gamma '}_d} = {{\gamma '}_s}\frac{{\sigma _{s,{r_{{\rm{min}}}}}^2}}{{\sigma _{d,{r_{{\rm{min}}}}}^2}} \triangleq {{\rho '}_d}{\gamma _u}.
$
Note that $\rho'_s$ and $\rho'_d$ are only associated with statistical channel conditions.

{\bf Case 2}: $\frac{{{\gamma _d}}}{{{\gamma _s}}} \ne \frac{{\sigma _{s,{r_{\text{min}}}}^2}}{{\sigma _{d,{r_{\text{min}}}}^2}}$.

In this case, the power allocation is represented by $\{\gamma''_s, \gamma''_d \}$, which is expressed based on the expressions of  $\{{\mathbb P}''_s, {\mathbb P''}_d\}$ as
\begin{align}
{{\gamma ''}_s} = \frac{{\left( {g' - 1} \right)\sigma _{u,v}^2{\gamma _u}}}{{{\Delta _u}\sigma _{s,v}^2\sqrt {g'\frac{{{\Delta _d}\sigma _{s,{r_{{\rm{min}}}}}^2\sigma _{d,v}^2}}{{{\Delta _s}\sigma _{d,{r_{{\rm{min}}}}}^2\sigma _{s,v}^2}} + 1} }} \triangleq {{\rho ''}_s}{\gamma _u}
\end{align}
and
\begin{align}
{{\gamma ''}_d} = \frac{{\left( {g' - 1} \right)\sigma _{u,v}^2{\gamma _u}}}{{{\Delta _u}\sigma _{d,v}^2\sqrt {g'\frac{{{\Delta _s}\sigma _{d,{r_{{\rm{min}}}}}^2\sigma _{s,v}^2}}{{{\Delta _d}\sigma _{s,{r_{{\rm{min}}}}}^2\sigma _{d,v}^2}} + 1} }} \triangleq {{\rho ''}_d}{\gamma _u},
\end{align}
where $\rho''_s$ and $\rho''_d$ are also associated with statistical channel conditions.

\bibliographystyle{IEEEtran}
\bibliography{IEEEabrv,References}

\begin{thebibliography}{10}
\providecommand{\url}[1]{#1}
\csname url@samestyle\endcsname
\providecommand{\newblock}{\relax}
\providecommand{\bibinfo}[2]{#2}
\providecommand{\BIBentrySTDinterwordspacing}{\spaceskip=0pt\relax}
\providecommand{\BIBentryALTinterwordstretchfactor}{4}
\providecommand{\BIBentryALTinterwordspacing}{\spaceskip=\fontdimen2\font plus
\BIBentryALTinterwordstretchfactor\fontdimen3\font minus
  \fontdimen4\font\relax}
\providecommand{\BIBforeignlanguage}[2]{{%
\expandafter\ifx\csname l@#1\endcsname\relax
\typeout{** WARNING: IEEEtran.bst: No hyphenation pattern has been}%
\typeout{** loaded for the language `#1'. Using the pattern for}%
\typeout{** the default language instead.}%
\else
\language=\csname l@#1\endcsname
\fi
#2}}
\providecommand{\BIBdecl}{\relax}
\BIBdecl

\bibitem{haykin2005cognitive}
S.~Haykin, ``Cognitive radio: brain-empowered wireless communications,''
  \emph{{IEEE} J. Sel. Areas Commun.}, vol.~23, no.~2, pp. 201--220, 2005.

\bibitem{jing2006distributed}
Y.~Jing and B.~Hassibi, ``Distributed space-time coding in wireless relay
  networks,'' \emph{{IEEE} Trans. Wireless Commun.}, vol.~5, no.~12, pp.
  3524--3536, 2006.

\bibitem{4786524}
Y.~Han, S.~H. Ting, C.~K. Ho, and W.~H. Chin, ``Performance bounds for two-way
  amplify-and-forward relaying,'' \emph{{IEEE} Trans. Wireless Commun.},
  vol.~8, no.~1, pp. 432--439, Jan 2009.

\bibitem{hunter2006outage}
T.~E. Hunter, S.~Sanayei, and A.~Nosratinia, ``Outage analysis of coded
  cooperation,'' \emph{{IEEE} Trans. Inf. Theory}, vol.~52, no.~2, pp.
  375--391, 2006.

\bibitem{zou2010adaptive}
Y.~Zou, J.~Zhu, B.~Zheng, and Y.-D. Yao, ``An adaptive cooperation diversity
  scheme with best-relay selection in cognitive radio networks,'' \emph{{IEEE}
  Trans. Signal Process.}, vol.~58, no.~10, pp. 5438--5445, 2010.

\bibitem{laneman2004cooperative}
J.~N. Laneman, D.~N. Tse, and G.~W. Wornell, ``Cooperative diversity in
  wireless networks: Efficient protocols and outage behavior,'' \emph{{IEEE}
  Trans. Inf. Theory}, vol.~50, no.~12, pp. 3062--3080, 2004.

\bibitem{rankov2007spectral}
B.~Rankov and A.~Wittneben, ``Spectral efficient protocols for half-duplex
  fading relay channels,'' \emph{{IEEE} J. Sel. Areas Commun.}, vol.~25, no.~2,
  pp. 379--389, 2007.

\bibitem{jing2009relay}
Y.~Jing, ``A relay selection scheme for two-way amplify-and-forward relay
  networks,'' in \emph{Wireless Communications \& Signal Processing, 2009. WCSP
  2009. International Conference on}.\hskip 1em plus 0.5em minus 0.4em\relax
  IEEE, 2009, pp. 1--5.

\bibitem{li2010asymptotically}
C.~Li, L.~Yang, and Y.~Shi, ``An asymptotically optimal cooperative relay
  scheme for two-way relaying protocol,'' \emph{{IEEE} Signal Process. Lett.},
  vol.~17, no.~2, pp. 145--148, 2010.

\bibitem{krikidis2010relay}
I.~Krikidis, ``Relay selection for two-way relay channels with mabc df: A
  diversity perspective,'' \emph{{IEEE} Trans. Veh. Technol.}, vol.~59, no.~9,
  pp. 4620--4628, 2010.

\bibitem{talwar2011joint}
S.~Talwar, Y.~Jing, and S.~Shahbazpanahi, ``Joint relay selection and power
  allocation for two-way relay networks,'' \emph{{IEEE} Signal Process. Lett.},
  vol.~18, no.~2, pp. 91--94, 2011.

\bibitem{6403863}
X.~Liang, S.~Jin, X.~Gao, and K.~K. Wong, ``Outage performance for
  decode-and-forward two-way relay network with multiple interferers and noisy
  relay,'' \emph{{IEEE} Trans. Commun.}, vol.~61, no.~2, pp. 521--531, February
  2013.

\bibitem{zhang2015exact}
X.~Zhang, Z.~Zhang, J.~Xing, R.~Yu, P.~Zhang, and W.~Wang, ``Exact outage
  analysis in cognitive two-way relay networks with opportunistic relay
  selection under primary user's interference,'' \emph{{IEEE} Trans. Veh.
  Technol.}, vol.~64, no.~6, pp. 2502--2511, 2015.

\bibitem{ubaidulla2012optimal}
P.~Ubaidulla and S.~Aissa, ``Optimal relay selection and power allocation for
  cognitive two-way relaying networks,'' \emph{{IEEE} Wireless Commun. Lett.},
  vol.~1, no.~3, pp. 225--228, 2012.

\bibitem{alsharoa2013relay}
A.~Alsharoa, F.~Bader, and M.-S. Alouini, ``Relay selection and resource
  allocation for two-way {DF-AF} cognitive radio networks,'' \emph{{IEEE}
  Wireless Commun. Lett.}, vol.~2, no.~4, pp. 427--430, 2013.

\bibitem{alsharoa2014optimal}
A.~Alsharoa, H.~Ghazzai, and M.-S. Alouini, ``Optimal transmit power allocation
  for mimo two-way cognitive relay networks with multiple relays using af
  strategy,'' \emph{{IEEE} Wireless Commun. Lett.}, vol.~3, no.~1, pp. 30--33,
  2014.

\bibitem{liu2015relay}
Y.~Liu, L.~Wang, T.~T. Duy, M.~Elkashlan, and T.~Q. Duong, ``Relay selection
  for security enhancement in cognitive relay networks,'' \emph{{IEEE} Wireless
  Commun. Lett.}, vol.~4, no.~1, pp. 46--49, 2015.

\bibitem{ahmed2004outage}
N.~Ahmed, M.~A. Khojastepour, and B.~Aazhang, ``Outage minimization and optimal
  power control for the fading relay channel,'' in \emph{Information Theory
  Workshop, 2004. IEEE}.\hskip 1em plus 0.5em minus 0.4em\relax IEEE, 2004, pp.
  458--462.

\bibitem{dulek2013optimum}
B.~Dulek, N.~D. Vanli, S.~Gezici, and P.~K. Varshney, ``Optimum power
  randomization for the minimization of outage probability,'' \emph{{IEEE}
  Trans. Wireless Commun.}, vol.~12, no.~9, pp. 4627--4637, 2013.

\bibitem{huang2014optimal}
C.~Huang, R.~Zhang, and S.~Cui, ``Optimal power allocation for outage
  probability minimization in fading channels with energy harvesting
  constraints,'' \emph{{IEEE} Trans. Wireless Commun.}, vol.~13, no.~2, pp.
  1074--1087, 2014.

\bibitem{ju2009catching}
H.~Ju, E.~Oh, and D.~Hong, ``Catching resource-devouring worms in
  next-generation wireless relay systems: Two-way relay and full-duplex
  relay,'' \emph{{IEEE} Commun. Mag.}, vol.~47, no.~9, pp. 58--65, 2009.

\bibitem{kim2011achievable}
S.~J. Kim, N.~Devroye, P.~Mitran, and V.~Tarokh, ``Achievable rate regions and
  performance comparison of half duplex bi-directional relaying protocols,''
  \emph{{IEEE} Trans. Inf. Theory}, vol.~57, no.~10, pp. 6405--6418, 2011.

\bibitem{zou2013optimal}
Y.~Zou, X.~Wang, and W.~Shen, ``Optimal relay selection for physical-layer
  security in cooperative wireless networks,'' \emph{{IEEE} J. Sel. Areas
  Commun.}, vol.~31, no.~10, pp. 2099--2111, 2013.

\end{thebibliography}

%







\begin{IEEEbiography}[{\includegraphics[width=1in,height=1.25in,clip,keepaspectratio]{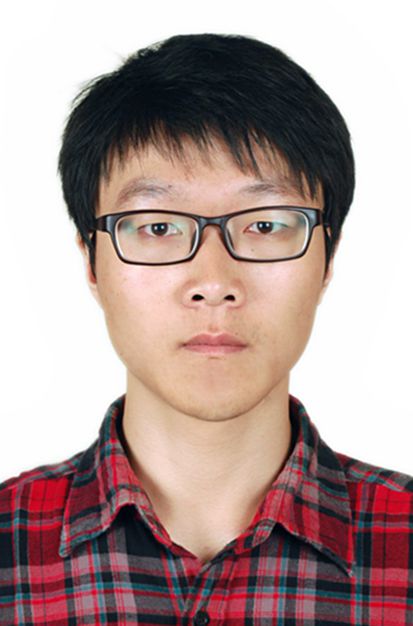}}]{Qunwei~Li}
	(S'16) received the B.S.\ and M.S.\ degrees in electrical engineering with honors from Xidian University, Xi'an, China, in 2011 and 2014. He has been pursuing the Ph.D.\ degree in the Department of Electrical Engineering and Computer Science, Syracuse University since 2014. 
	
	He received the Syracuse University Graduate Fellowship award in 2014. His research interests include statistical signal processing, crowdsourcing, machine learning, and optimization. 
\end{IEEEbiography}

\begin{IEEEbiography}[{\includegraphics[width=1in,height=1.25in,clip,keepaspectratio]{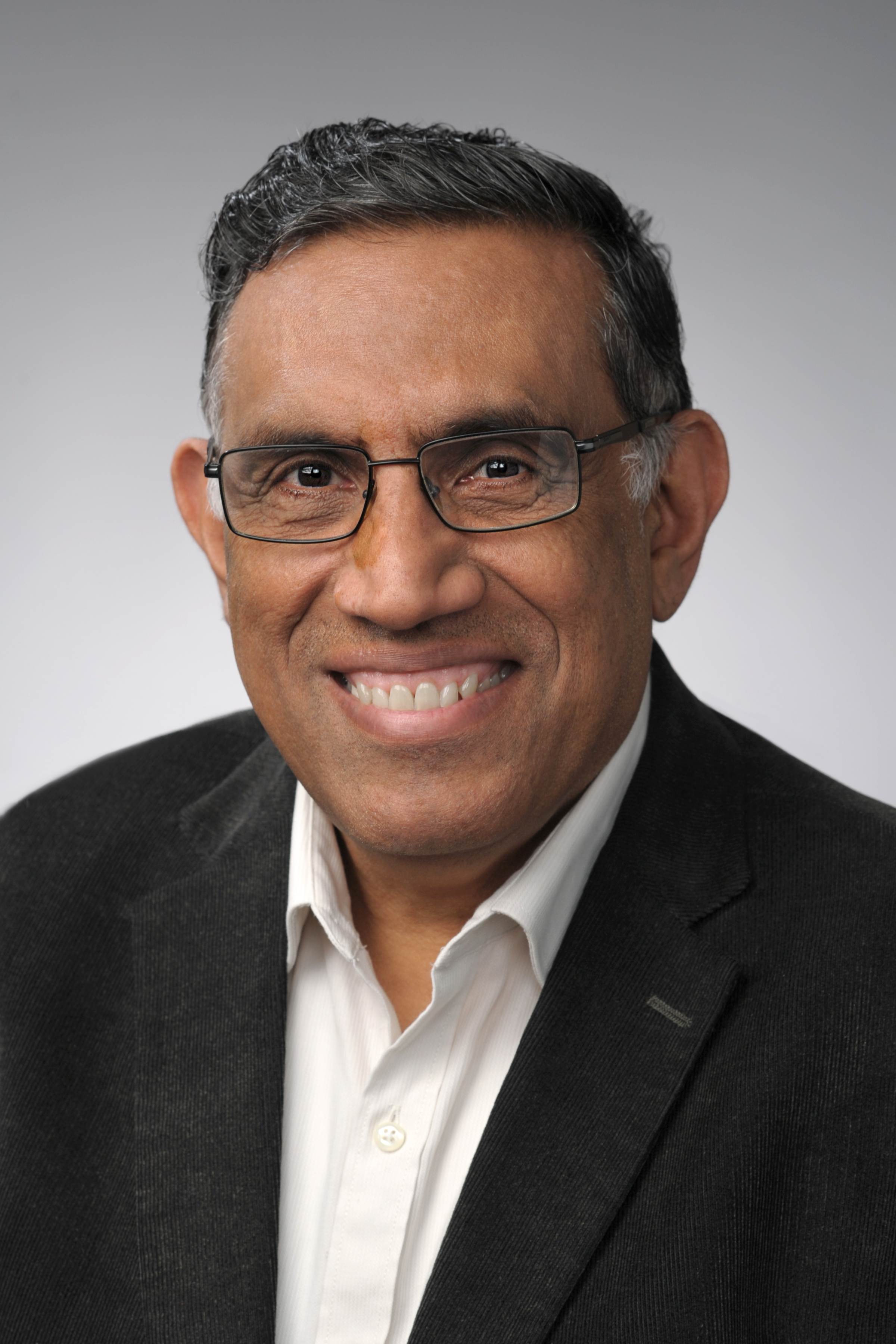}}]{Pramod~K.~Varshney}
	(S'72-M'77-SM'82-F'97) was born in Allahabad, India, on July 1, 1952. He received the B.S.\ degree in electrical engineering and computer science (with highest honors), and the M.S. and Ph.D. degrees in electrical engineering from the University of Illinois at Urbana-Champaign, USA, in 1972, 1974, and 1976 respectively.
	
	From 1972 to 1976, he held teaching and research assistantships with the University of Illinois. Since 1976, he has been with Syracuse University, Syracuse, NY, where he is currently a Distinguished Professor of Electrical Engineering and Computer Science and the Director of CASE: Center for Advanced Systems and Engineering. He served as the associate chair of the department from 1993 to 1996. He is also an Adjunct Professor of Radiology at Upstate Medical University, Syracuse. His current research interests are in distributed sensor networks and data fusion, detection and estimation theory, wireless communications, image processing, radar signal processing, and remote sensing. He has published extensively. He is the author of Distributed Detection and Data Fusion (New York: Springer-Verlag, 1997). He has served as a consultant to several major companies. 
	
	Dr.\ Varshney was a James Scholar, a Bronze Tablet Senior, and a Fellow while at the University of Illinois. He is a member of Tau Beta Pi and is the recipient of the 1981 ASEE Dow Outstanding Young Faculty Award. He was elected to the grade of Fellow of the IEEE in 1997 for his contributions in the area of distributed detection and data fusion. He was the Guest Editor of the Special Issue on Data Fusion of the Proceedings of the IEEE January 1997. In 2000, he received the Third Millennium Medal from the IEEE and Chancellor's Citation for exceptional academic achievement at Syracuse University. He is the recipient of the IEEE 2012 Judith A. Resnik Award, the degree of Doctor of Engineering honoris causa by Drexel University in 2014 and the ECE Distinguished Alumni Award from UIUC in 2015. He is on the Editorial Boards of the Journal on Advances in Information Fusion and IEEE Signal Processing Magazine. He was the President of International Society of Information Fusion during 2001.
\end{IEEEbiography}

\end{document}